\documentclass[11pt]{article}

\usepackage{amsmath}
\usepackage{amssymb}
\usepackage{amsthm}
\usepackage{tikz}
\usepackage{graphicx}
\usepackage{float}
\usepackage{algorithm}
\usepackage[noend]{algpseudocode}
\usepackage{tabularx}
\usepackage[margin=1in]{geometry}
\usepackage{arydshln}
\usepackage{multirow}
\usepackage{wrapfig}

\usetikzlibrary{positioning,calc,shapes,arrows}
\usetikzlibrary{backgrounds}
\usetikzlibrary{matrix,shadows,arrows} 
\usetikzlibrary{decorations.pathreplacing}

\def\etal.{et\penalty50\ al.}
\theoremstyle{plain}
\newtheorem{theorem}{Theorem}[section]
\newtheorem{lemma}[theorem]{Lemma}

\newtheorem{claim}[theorem]{Claim}

\theoremstyle{definition}
\newtheorem{definition}{Definition}[section]

\theoremstyle{remark}
\newtheorem{question}{Question}[section]

\newtheorem{remark}[question]{Remark}
\newtheorem{openproblem}[question]{Open Problem}
\theoremstyle{plain}
\newtheorem*{theorem*}{Theorem}

\DeclareMathOperator{\OPT}{OPT}
\DeclareMathOperator{\ALG}{ALG}
\DeclareMathOperator{\AR}{AR}

\DeclareMathOperator{\Ima}{Im}

\DeclareMathOperator*{\argmin}{arg\,min}

\DeclareMathOperator{\curdeg}{cur\,deg}

\title{On Conceptually Simple Algorithms for Variants of Online Bipartite Matching}
\author{Allan Borodin\thanks{Research is supported by NSERC.} \\ University of Toronto \\ \textsf{bor@cs.toronto.edu}  \and
Denis Pankratov\footnotemark[1] \\ University of Toronto \\ \textsf{denisp@cs.toronto.edu} \and
Amirali Salehi-Abari\footnotemark[1] \\ University of Toronto \\ \textsf{abari@cs.toronto.edu}}

\begin{document}

\maketitle

\begin{abstract}
We present a series of results regarding 
conceptually simple algorithms for bipartite matching in various online and related models. 

We  first consider a deterministic adversarial model. The best approximation ratio possible for a one-pass deterministic online algorithm  is $1/2$, which is achieved by any  greedy algorithm. D\"urr \etal. \cite{Durr2016} recently presented a $2$-pass algorithm called \textsc{Category-Advice} that achieves approximation ratio $3/5$. We extend their algorithm to multiple passes. We prove the exact approximation ratio for the \textsc{$k$-pass Category-Advice} algorithm for all $k \ge 1$, and show that the approximation ratio converges to the inverse of the golden ratio $2/(1+\sqrt{5}) \approx 0.618$ as $k$ goes to infinity. The convergence is extremely fast --- the \textsc{$5$-pass Category-Advice} algorithm is already within $0.01\%$ of the inverse of the golden ratio.

We then consider a natural greedy algorithm in the online stochastic IID model -- 
\textsc{MinDegree}. This algorithm is an online version of a well-known and extensively studied offline algorithm \textsc{MinGreedy}. \textsc{MinGreedy} is known to have excellent empirical performance, as well as excellent performance under various input models. Alas, recently it was shown to have approximation ratio $1/2$ in the adversarial offline setting --- the approximation ratio achieved by any greedy algorithm. We obtain a result in the online known IID model that, in spirit, is similar to the offline result, but the proof is different. Namely, we show that \textsc{MinDegree} cannot achieve an approximation ratio better than $1-1/e$, which is guaranteed by any consistent greedy algorithm in the known IID model.

Finally, following the work in Besser and Poloczek \cite{BesserP17}, we depart from an adversarial or stochastic ordering and investigate a natural randomized algorithm (\textsc{MinRanking}) in the priority model. Although the priority model allows the algorithm to choose the input ordering in a general but well defined way, this natural algorithm cannot obtain the approximation of the \textsc{Ranking} algorithm in the ROM model. 
\end{abstract}

\section{Introduction}
\label{sec:intro}
Maximum bipartite matching (MBM) is a classic extensively studied graph problem. Let $G=(U,V,E)$ be a bipartite graph, where $U$ and $V$ are the vertices on the two sides, and $E \subseteq U \times V$ is a set of $m$ edges. In 1931, K\"onig and Egerv\'{a}ry independently gave a characterization of MBM in terms of the minimum vertex cover. The celebrated Hopcroft-Karp algorithm~\cite{HopcroftKarp1973} running in time $O(m\sqrt{n})$ was discovered in 1973 where $n$ is the number of vertices. The first improvement in the regime of relatively sparse graphs came forty years later, in 2013, when Madry~\cite{Madry2013} developed a $\widetilde{O}(m^{10/7})$ algorithm based on electrical flows. For dense graphs, i.e., when $m \approx n^2$, Mucha and Sankowski~\cite{MuchaSankowski2004} describe an algorithm running in time $O(n^\omega)$, where $\omega \le 2.373$ is the matrix multiplication constant.  We refer the interested reader to \cite{Madry2016} and \cite{DuanP14} and references therein for more information on MBM in the offline setting. 
While current algorithms for solving MBM optimally in the offline setting are reasonably efficient, they still fall short of linear time algorithms. For large graphs, linear or near linear time algorithms might be required. In that regard, a $(1-\epsilon)$-approximation can be computed in $O(m/\epsilon)$ time by a version of the Hopcroft-Karp algorithm in the offline setting ~\cite{DuanP14}. 
Arguably, such algorithms are not that conceptually simple and require a reasonable understanding of the problem. We now outline the computational models relevant to our focus on simple algorithms, and then highlight our results.     

\paragraph{Conceptually Simple Algorithms.} Our focus in this paper is on ``conceptually simple algorithms'' and, in particular, on such algorithms with regard to variants of the online bipartite matching problem. We will not define ``conceptual simplicity'' but claim that certain types of algorithms  (e.g., greedy and  local search) usually fall within this informal ``you know it when you see it'' concept. The online model is sometimes necessitated by  applications, and can be studied  
 with respect to a completely adversarial model, the random order  model (ROM), or a distributional input model (e.g., known and unknown IID input models). In all of these models, the algorithm has no control over the ordering of the inputs and 
must make irrevocable decisions for each input item as it arrives. As such, online algorithms are a prime example of a conceptually simple algorithmic paradigm that can be extended in various ways leading to simple offline algorithms. These online extensions can provide much improved performance both in terms of worst-case approximation ratios and in terms of performance on real data. See, for example, the experimental analysis of MaxSat provided by Poloczek and Williamson \cite{PoloczekW16}. However, this 
still begs the question as to why we should restrict ourselves to conceptually simple algorithms when much better offline algorithms are known. Indeed as already stated, there are 
polynomial time offline optimal algorithms and efficient FPTAS algorithms running in time $O(m/\epsilon)$ for bipartite matching.  
  
While conceptually simple  algorithms may not match the best approximations realized by more complex methods, they are usually very efficient (i.e., linear or near linear time with small constant factors) and often work well on realistic data exceeding 
worst-case approximation bounds. Conceptually simple algorithms can also 
be used as a 
preprocessing step for initializing a local search algorithm as in 
Chandra and 
Halld{\'{o}}rsson \cite{ChandraH01a}. Moreover, and this 
can be even more critical in many settings, conceptually simple algorithms are easy to implement and modify with relatively little knowledge about the problem domain. Indeed, conceptual simplicity is arguably the main reason for the use of  simple mechanisms in auctions (see, for example, Lucier and Syrgkanis \cite{LucierS15}) and the success of MapReduce in distributed parallel applications as introduced by Dean and Ghemawat \cite{DeanG08}.  

We will consider two departures from the adversarial and distributional one-pass online models. In the first departure, we consider a multi-pass online algorithm generalizing the two-pass algorithm in D{\"u}rr \etal.\cite{Durr2016}. In this regard we are also motivated by the Poloczek et al. \cite{PoloczekSWZ17} two-pass algorithm 
for MaxSat.  The D{\"u}rr \etal. two-pass algorithm and our extension to a $k$-pass online bipartite algorithm can also be viewed as an $O(n)$ space 
semi-streaming algorithm in the vertex input model. We can also view these multi-pass algorithms as de-randomizations of randomized online algorithms. 

The second departure is that of priority algorithms \cite{BorodinNR2003}, a model for greedy and more generally myopic algorithms that extend online algorithms by allowing the algorithm to determine (in some well-defined way) the order of input arrivals. 

Other conceptually simple  generalizations of the traditional online model are clearly possible, such as the ability to modify previous decision (e.g., as in \cite{EpsteinLSW13}) and parallel executions of online algorithms (e.g., as in \cite{Hall2002, Buchbinder2016, PenaB16}).   

\paragraph{Adversarial Online Model.} In 1990, Karp, Vazirani, and Vazirani~\cite{KarpVV90} initiated the study of MBM in the online setting. In this setting, nodes $V$ are the offline nodes known to an algorithm beforehand, and nodes $U$ arrive online in some adversarial order. When a node in $U$ arrives, all its neighbors in $V$ are revealed simultaneously. An online algorithm is required to make an irrevocable decision with regard to which neighbor (if any) the arrived node is matched to. 
Any greedy algorithm (yielding a maximal matching) achieves a $1/2$ approximation and 
Karp, Vazirani, and Vazirani showed that no deterministic algorithm can achieve an (asymptotic) approximation ratio better than $1/2$ in the adversarial online model. They gave a randomized online algorithm called \textsc{Ranking} and showed that it achieves a $1-1/e \approx 0.632$ expected approximation ratio. Moreover, they proved that no randomized algorithm can beat $1-1/e$ in the adversarial online model. Seventeen years after the publication of the \textsc{Ranking} algorithm a mistake was found in the analysis of the Ranking algorithm. The mistake was discovered 
independently by Krohn and Varadarajan and by Goel and Mehta in \cite{GoelM2008}, 
and a correct proof was provided by Goel and Mehta. 
Since then many different and simplified proofs of the approximation ratio of \textsc{Ranking} have been given (see \cite{GoelM2008,BirnbaumM2008,DevanurJK2013}). Thus, the one-pass  adversarial online setting for MBM is now reasonably well understood.

\paragraph{Online Stochastic Models.} Feldman \etal. \cite{FeldmanMVM2009} introduced the known IID distributional model for MBM, which is motivated by practical considerations. In the applications of MBM to online advertising, one often knows some statistics about the upcoming queries (online nodes). Feldman \etal. model this by the notion of a type graph $G = (U, V, E)$ and a probability distribution $p : U \rightarrow [0,1]$. The online nodes are sampled from $p$ independently one at a time. An algorithm knows $G$ and $p$ beforehand. As before, an algorithm is required to make an irrevocable decision about which neighbor to match the newly arriving online node to. In this setting, the adversary can only choose the type graph and the distribution $p$ but doesn't have further control over the online sequence of nodes, as those are sampled in the IID fashion. Thus, the adversary is more restricted than in the adversarial online model. Feldman \etal. \cite{FeldmanMVM2009} describe an algorithm beating the $1-1/e$ barrier and achieving approximation ratio $\approx 0.67$. This work was followed by a long line of work including   ~\cite{Bahmani2010,Manshadi2011,HaeuplerMZ2011,Jaillet2014,BrubachSSX2016}. So far, the best approximation ratio for arbitrary arrival rates is $\approx 0.706$ due to Jaillet and Lu \cite{Jaillet2014}. 

Other online stochastic input models have been studied; e.g., Known IID with integer types, Random Order Model (ROM), and Unknown IID. In addition, practical considerations have motivated a study of MBM-like problems
and extensions that include Vertex-Weighted Matching, Edge-Weighted Matching, Adwords, DisplayAds, Online Submodular Welfare Maximization, among the most well-known. For more information and background on these models and problems, we refer the reader to the excellent survey by Mehta \cite{Mehta2013}. The main reason for such an abundance of models and problems is that in recent years the online advertising industry (and more generally
online auctions) has grown extensively. Online advertising platforms essentially solve various flavors of MBM. Online MBM is an important example of a problem that is at the intersection of practice and theory.

\paragraph{Semi-streaming Model.}

One-pass and multi-pass streaming algorithms are important algorithmic models motivated by the 
necessity to process extremely large data streams where the amount of data may preclude 
storing all of the input and hence small 
space bounds are imposed throughout the computation. 
The usual assumption is that the data stream is created by an adversary but 
other input streams (e.g., distributional) are also possible. 
Much of the streaming literature concerns various forms of counting and statistics gathering. Semi-streaming algorithms are streaming algorithms designed for (say) 
graph search and optimization problems  
where the output itself requires $O(n)$ space 
and hence a realistic goal is to 
maintain ${\tilde O}(n)$ space rather than space $O(m)$. 
The semi-streaming model (with edge inputs) was introduced by  Feigenbaum \etal. \cite{FeigenbaumKMSZ05} and subsequently studied in \cite{EpsteinLMS11, KonradMM12, EggertKMS12}. In particular, Eggert \etal. \cite{EggertKMS12} provide a FPTAS multi-pass semi-streaming algorithm for MBM using space ${\tilde O}(n)$ in the edge input model. In the vertex input semi-streaming model, Goel \etal. \cite{GoelKK12} give a {\it deterministic} $1-1/e$ approximation and Kapralov \cite{Kapralov13} proves that no semi-streaming algorithm can improve upon this ratio. (See also the recent survey by McGregor \cite{McGregor17}.)
The difference between semi-streaming
algorithms and online algorithms in the sense of competitive analysis is that 
streaming algorithms do not have to make online decisions but must maintain 
small space throughout the computation while online algorithms must make 
irrevocable decisions for each input item but have no  space requirement. The Goel \etal. result shows the power of deterministic semi-streaming over deterministic online algorithms. 
In some cases, streaming algorithms are designed so as
to make results available                                   
at any time (after each input item) during the computation and hence some streaming 
algorithms 
can also be viewed both as a streaming algorithm and an online algorithm. 
Conversely, any algorithm that restricts itself to ${\tilde O}(n)$ 
space can also be considered as a semi-streaming algorithm.  

\paragraph{Priority Model.} The priority model~\cite{BorodinNR2003} captures the notion of greedy-like algorithms. In this setting, an algorithm has some control over the order of online nodes. More specifically, an input to the algorithm is represented as a set of input items coming from some (possibly infinite) universe. For example, the universe of items in the MBM problem is the set of all pairs $(u, N(u))$ where $u$ is an online node and $N(u)$ is its neighborhood, and an instance is a collection of such pairs specifying the graph $G$. The algorithm first commits to an ordering of the entire universe.\footnote{In practice, the ordering is determined by a priority function that maps each possible input item into a real number and then the ordering is determined by the priority values.} Then, the adversary picks an instance $G$ and reveals the online nodes in the order specified by the algorithm. The algorithm is still required to make an irrevocable decision about a newly arriving node. This captures the idea that many offline greedy algorithms sort the input items, and then do a single pass over the items in the sorted order. More generally, priority algorithms can adaptively reorder so as to select the next input item to process.  Many problems have been studied in the priority model \cite{AngelopoulosB2002, BorodinBL2005, YeB08, PoloczekD2011, DavisI2009, BorodinIYZ2012, BesserP17}. The original deterministic priority model was extended to the randomized priority model in \cite{AngelopoulosB2010}. We shall use the term {\it fully randomized priority algorithm} to indicate that the ordering of the input items and the decisions for each item are both randomized. When only the decisions are randomized (and the ordering is deterministic) we will simply say {\it randomized priority algorithm}. With regards to maximum matching, Aronson \etal. \cite{AronsonDFS1995} proved that an algorithm that picks a random vertex and matches it to a random available neighbor (if it exists) achieves approximation ratio $1/2 + \epsilon$ for some $\epsilon > 0$ in general graphs. Besser and Poloczek ~\cite{BesserP17} show that the algorithm that picks a random vertex of minimum degree and matches it to a randomly selected neighbor cannot 
improve upon the $1/2$ approximation ratio (with high probability) even for bipartite graphs. Pena and Borodin \cite{PenaB16} show that no deterministic (respectively, fully randomized) priority algorithm can achieve approximation ratio better than $1/2$ 
(respectively 
$53/54$)  for the  MBM problem. (See also \cite{BesserP2017} with respect to the  
the difficulty of proving inapproximation results for all randomized 
priority algorithms.)

\paragraph{Advice Model.} D\"urr \etal. in \cite{Durr2016} studied the online MBM problem in the adversarial (tape) \emph{advice model}. Advice can be viewed as a generalization of randomness. A randomized algorithm is given random bits (independent of the input instance), whereas an advice algorithm is given advice bits prior to processing the online input. In the most unrestricted advice setting, the advice bits are set by an all-powerful oracle.  
D\"urr \etal. show that $\Theta_\epsilon(n)$ advice bits are necessary and sufficient \footnote{More precisely, $O(\frac{n}{\epsilon^5})$ advice bits are sufficient and 
$\Omega(\log( \frac{1}{\epsilon})n)$ advice bits are necessary to achieve a $(1-\epsilon)$ approximation. The approximation algorithm is based on the multi-pass semi-streaming result of Eggert \etal. \cite{EggertKMS12} and lacks the simplicity of the
D\"urr \etal. \cite{Durr2016} \textsc{Category Advice} algorithm.}  to guarantee approximation ratio $1-\epsilon$ for MBM. They also show that $O(\log n)$ advice bits are sufficient for a deterministic advice algorithm to guarantee a $1-1/e$ approximation ratio. (This result is based on the randomization to advice transformation due to B\"ockenhauer \etal. \cite{BockenhauerKKK2017}.) Construction of the $O(\log n)$ advice bits is based on examining the behavior of the \textsc{Ranking} algorithm on all $n!$ possible random strings for a given input of length $n$, which requires exponential time. It is not known if there is an efficient way to construct $O(\log n)$ advice bits. More generally, one may put computational or information-theoretic restrictions on the advice string, and ask what approximation ratios are achievable by online algorithms with restricted advice. Not only is this framework of theoretical value, but it also gives rise to classes of conceptually simple offline algorithms if the advice string is restricted to be efficiently computable. 
D\"urr \etal. \cite{Durr2016} present a deterministic advice algorithm \textsc{Category-Advice} achieving approximation ratio $3/5$ with an $n$-bit advice string, where the advice string itself is computable by an online algorithm. This algorithm can obviously be viewed as a 2-pass online algorithm. 
\subsection{Summary of Our Results}
In this subsection, we briefly describe our results on conceptually simple algorithms under the input models discussed above.

\paragraph{Our Online Multi-pass Results.}
We generalize the \textsc{Category-Advice} algorithm to a  \textsc{$k$-pass Category-Advice} algorithm for $k \ge 1$. For each $k \ge 1$, we prove that the exact approximation ratio of \textsc{$k$-pass Category-Advice} algorithm is $F_{2k}/F_{2k+1}$, where $F_n$ is the $n$th Fibonacci number. Our bounds show that the analysis of D\"urr \etal. for the \textsc{$2$-pass Category-Advice} algorithm was tight. Our result immediately implies that the approximation ratio of \textsc{$k$-pass Category-Advice} converges to the inverse of the golden ratio $2/(1+\sqrt{5}) \approx 0.618$ as $k$ goes to infinity.

\paragraph{Our Results for the Known IID Model.} A greedy algorithm always matches an online vertex if it has at least one available neighbor. We observe that every online algorithm for bipartite matching  
can be converted into a greedy online algorithm without hurting its approximation ratio. The greedy property does not specify how to break ties when several neighbors of an online vertex are available. Goel and Mehta~\cite{GoelM2008} proved that most greedy algorithms (i.e., no matter what tie-breaking rules are used as long as they are consistent) achieve at least a $1-1/e$ approximation ratio in the known IID model. As mentioned earlier, algorithms achieving better approximation ratios are known. They are usually stated as non-greedy algorithms, but using a general conversion they can be turned into greedy algorithms, albeit with somewhat unnatural tie-breaking rules. These algorithms have polynomial time preprocessing of the type graph and, thus, are feasible from a theoretical point of view. The preprocessing step, although polynomial time, might not be feasible from a practical point of view on large type graphs. Moreover, we argue that these algorithms are not that ``conceptually simple.'' At present, the research literature on conceptually simple tie-breaking rules for greedy algorithms in the known IID model is scarce. We introduce and study a deterministic greedy online algorithm arising out of a natural, conceptually simple, and efficient tie-breaking rule  --- \textsc{MinDegree}. This algorithm is motivated by a well-known offline matching algorithm  --- \textsc{MinGreedy \cite{Tinhofer84}}. We show that \textsc{MinDegree} does not beat the $1-1/e$ approximation achieved by any greedy algorithm (using consistent tie-breaking).

\paragraph{Our Results for the Priority Model.}
When we consider the \textsc{Ranking} algorithm in the ROM model, we have an instance
of a priority algorithm where both the order of input arrivals and the 
decisions made for each online input are randomized. This ordering is simply 
a uniform random permutation of the set of adversarial chosen input items. 
Is there a more ``informed'' way to deterministically or randomly choose 
the ordering within the priority framework? A natural idea is to give priority to the online nodes having the smallest degree since intuitively they would seem to 
be the hardest to match if not seen early. Using the Ranking algorithm to match 
online nodes, our intuition turns out to 
be misleading. We show that giving priority to nodes having the smallest degree using some deterministic (or the uniform random tie-breaking) rule cannot match the approximation achieved by a uniform ordering of the online nodes. In contrast to the \textsc{$2$-pass Category-Advice} $3/5$ approximation, 
our analysis can also be used to show that a deterministic two-pass algorithm that computes  the degree of the offline vertices in the first pass and then reorders the offline vertices  according to non-decreasing degree (breaking ties by the initial ordering) will not achieve an asymptotic approximation ratio better than $1/2$.

\paragraph{Organization.} The rest of this paper is organized as follows. Section~\ref{sec:prelim} introduces key definitions, notation, and a statement of the \textsc{Ranking} algorithm. In Section~\ref{sec:multipass} we describe the \textsc{$k$-pass Category-Advice} algorithm and prove the exact approximation ratio for each $k \ge 1$, giving matching upper and lower bounds. 
In Section~\ref{sec:stochastic}, we provide an analysis of \textsc{MinDegree} 
in the known IID model. In Section~\ref{sec:priority}, we describe the \textsc{MinRanking} algorithm and provide analysis of the family of graphs due to Besser and Poloczek \cite{BesserP17}. We conclude and present some open problems in Section~\ref{sec:conclusion}.

\section{Preliminaries}
\label{sec:prelim}
We are given a bipartite graph $G=(U,V,E)$ where $U$ and $V$ form a partition of the vertices, and  
edges are $E \subseteq U \times V$. Names of the vertices are natural numbers. We consider the maximum bipartite matching problem in various online models. We typically think of $U$ as the online vertices that are revealed to a given algorithm one at a time, and $V$ as the offline vertices. When a vertex $u$ from $U$ is revealed, all edges incident on $u$ are also revealed simultaneously. The online models differ in how vertices in $U$ (and the 
order of arrivals) are chosen 
--- adversarially, sampled from a known distribution (with a known type graph) in the IID fashion, or via a limited adversary (as in the priority model where the adversary chooses the graph but not the sequence of arrivals).

Let $M \subseteq E$ be some matching in $G$. For $u \in U$, we write $u \in M$ when there exists $v \in V$ such that $(u,v) \in M$. Similarly, we write $v \in M$ when there exists $u \in U$ such that $(u,v) \in M$. For $u \in M$, we write $M(u)$ to indicate the neighbor of $u$ in $M$. 
We write $\OPT(G)$ to stand for an offline optimum matching in $G$.

Given a deterministic algorithm $\ALG$ in any of the models being considered and a bipartite graph $G$, we let $\ALG(G)$ stand for the matching returned by the algorithm $\ALG$ on input $G$. Abusing notation, we will also use $\ALG(G)$ 
(resp. $\OPT(G)$) to stand for the size of the  matching returned by the algorithm on $G$ (resp. by $\OPT(G)$).

\begin{definition} 

For a deterministic algorithm $\ALG$ and adversarial input, we define its \emph{asymptotic approximation ratio} as
\[ \AR(\ALG) = \lim_{\OPT(G) \rightarrow \infty} \inf_{G} \frac{\ALG(G)}{\OPT(G)}.\]

For a deterministic algorithm and distributional input, its \emph{asymptotic approximation ratio} is
\[ \AR(\ALG) = \lim_{\mathbb{E}[\OPT(G)] \rightarrow \infty} \inf_{G} 
\frac{\mathbb{E}[\ALG(G)]}{\mathbb{E}[\OPT(G)]},\] 
where the expectation is over the input distribution.

For a randomized algorithm and adversarial input,  its \emph{asymptotic approximation ratio} is
\[ \AR(\ALG) = \lim_{\OPT(G) \rightarrow \infty} \inf_{G}
\frac{\mathbb{E}[\ALG(G)]}{\OPT(G)},\]
where the expectation is taken over the randomization 
in the algorithm.

\end{definition}

\begin{definition}
An online (or priority) MBM algorithm is greedy if whenever a newly arriving online node has at least one available neighbor the algorithm matches the arrived node to one of its neighbors.
\end{definition}

\begin{remark}\label{rem:greedywlog}
Any online (or priority) algorithm for MBM achieving approximation ratio $\rho$ (in an adversarial or stochastic input setting) can be turned into a \emph{greedy} algorithm achieving approximation ratio $\ge \rho$. Informally, the idea is a simulation of the non-greedy algorithm in which we  replace any non-greedy decision by forcing a match while remembering the configuration of the non-greedy algorithm.
\end{remark}

In the bipartite graphs $G=(U,V,E)$ that we consider in this paper, we shall often refer to the so-called ``parallel edges.'' Let $U' = \{u'_1, \ldots, u'_k\} \subseteq U$ and $V' = \{v'_1, \ldots, v'_k\} \subseteq V$ be two distinguished subsets such that for all $i \in [k]$ we have $(u'_i, v'_i) \in E$ (there might be other edges incident on $U'$ and $V'$). The parallel edges between $U'$ and $V'$ are precisely the edges $(u_i', v_i')$. 

\subsection{The \textsc{Ranking} Algorithm}
In the adversarial model, graph $G$ as well as an order $\pi$ of its online vertices $U$ is chosen by an adversary. Karp, Vazirani, and Vazirani \cite{KarpVV90} presented an optimal randomized algorithm in the adversarial model called \textsc{Ranking} (see Algorithm~\ref{algorithm:ranking}). They showed that $\AR(\textsc{Ranking}) = 1-1/e$ and that no randomized algorithm can do better in the adversarial model. \textsc{Ranking} works by sampling a permutation $\sigma$ of the offline nodes uniformly at random and running the natural simple greedy algorithm on $G$ breaking ties using $\sigma$. That is, when an online vertex $u$ arrives, if $u$ has many unmatched neighbors $v$ then $u$ is matched with $v$ that minimizes $\sigma(v)$. 

\begin{algorithm}
\caption{The \textsc{Ranking} algorithm.}\label{algorithm:ranking}
\begin{algorithmic}
\Procedure{Ranking}{$G=(U,V,E)$}
\State{Pick a permutation $\sigma : V\rightarrow V$ uniformly at random}
\ForAll{$u \in U$ }
\State{When $u$ arrives, let $N(u)$ be the set of unmatched neighbors of $u$}
\If{$N(u) \neq \emptyset$}
\State{Match $u$ with $\argmin_v \{\sigma(v) \mid v \in N(u)\}$}
\EndIf
\EndFor
\EndProcedure
\end{algorithmic}
\end{algorithm}

Let $\pi$ be a permutation of $U$ in which the nodes are revealed online, and $\sigma$ be a permutation of $V$ chosen by the algorithm. In this case, we write \textsc{Ranking}$(\pi, \sigma)$ to denote the matching returned by \textsc{Ranking} given $\pi$ and $\sigma$. When $\pi$ is clear from the context we will omit it and simply write \textsc{Ranking}$(\sigma)$.

\subsection{Known IID Model with Uniform Probability Distribution}

In the known IID model, graph $G = (U,V,E)$ is called a \emph{type graph}. Nodes in $u$ are sometimes referred to as ``types.'' The type graph specifies the distribution from which the actual \emph{instance graph} $\widehat{G}$ is generated. An instance graph $\widehat{G}=(\widehat{U}, \widehat{V}, \widehat{E})$ is generated by setting $\widehat{V} = V$, sampling each  $\widehat{u} \in \widehat{U}$ IID uniformly from $U$, and letting $(\widehat{u},\widehat{v}) \in \widehat{E}$ if and only if the corresponding $(u,v) \in E$.

Note that each online $\widehat{u}$ is drawn independently  from $U$ with replacement, thus it is possible that the same $u \in U$ appears multiple times in $\widehat{U}$. The type graph is chosen adversarially and is revealed to the algorithm in advance. After that, nodes $\widehat{u}$ are generated one by one on the fly and presented to the algorithm one at a time. An online algorithm makes an irrevocable decision on how to match a newly generated $\widehat{u}$ (if at all) without the knowledge of the future. We shall sometimes refer to $U$ and $\widehat{U}$ as the online side. Note that the known IID model can be defined more generally by allowing an arbitrary distribution $p$ on $U$ from which $\widehat{u}$ is sampled IID. We don't do that since our (negative) results only need the case when $p$ is the uniform distribution. We make one additional assumption that $|\widehat{U}| = |U|$. This is not necessary, but all our examples are in this setting.

For a given type graph $G$ we shall write $\mu_G$ to denote the corresponding distribution defined above. Observe that in this setting, we have $\OPT(G) = \mathbb{E}_{\widehat{G} \sim \mu_G} (M(\widehat{G}))$, and $\ALG(G) = \mathbb{E}_{\widehat{G} \sim \mu_G} (\ALG(\widehat{G}))$.

Algorithms can be taken to be greedy without loss of generality in the known IID model (see Remark~\ref{rem:greedywlog}). In addition, a basic guarantee on the performance of a greedy algorithm in the known IID model holds as long as its tie-breaking rules satisfy natural consistency conditions. The conditions say that if an online node $u$ is matched with $u^*$ in some run of the algorithm and $u^*$ is available in another ``similar'' run of the algorithm, $u$ should still be matched with $u^*$ (``similar'' means that neighbors of $u$ in the second run form a subset of the neighbors of $u$ in the first run). More formally:

\begin{definition}
\label{def:consistency}
Let $\ALG$ be a greedy algorithm. Let $\widehat{G}(\pi)$ denote a type graph with online nodes presented in the order given by $\pi$. Let $N_\pi(u)$ denote the set of neighbors of $u$ that are available to be matched when $u$ is processed by $\ALG$ running on $\widehat{G}(\pi)$. We say that a greedy algorithm is \emph{consistent or that it satisfies consistency conditions} if the following holds: $\forall \pi, \pi', u$ if $N_{\pi'}(u) \subseteq N_{\pi}(u)$ and $u$ is matched with $u^* \in N_{\pi'}(u)$ by $\ALG(\widehat{G}(\pi))$ then $u$ is also matched with $u^*$ by $\ALG(\widehat{G}(\pi'))$.
\end{definition}

Goel and Mehta~\cite{GoelM2008} proved that such greedy algorithms achieve approximation at least $1-1/e$ in the known IID model.\footnote{The Goel and Mehta~\cite{GoelM2008} result is even stronger as it holds for the ROM model, but we do not need the stronger result in our paper.}

\begin{theorem}[\cite{GoelM2008}]
\label{thm:known_iid_greedy_positive}
Every consistent greedy algorithm $\ALG$ in the known IID model satisfies
\[\AR(\ALG) \ge 1-1/e.\]
\end{theorem}

\subsection{Priority Model}

In the priority model, an input instance consists of several data items. For the MBM problem, a data item is a pair $(u, N(u))$, where $u$ is the name of an online vertex and $N(u)$ is the set of names of its offline neighbors. Let $\mathcal{I}$ be the infinite universal set of all data items. An input instance is then a finite subset $S = \{(u_i, N(u_i))\} \subseteq \mathcal{I}$ that defines the bipartite graph $G = (U,V,E)$ in a natural way: $U = \{u_i\}, V = \cup_i N(u_i),$ and $(u,v) \in E$ if and only if there exists $j$ such that $u = u_j$ and $v \in N(u_j)$. A \emph{fixed priority} algorithm a priori chooses a permutation $\tau : \mathcal{I} \rightarrow \mathcal{I}$ of the entire universe of data items. The adversary then chooses an instance $S \subseteq \mathcal{I}$ and presents it to the algorithm one data item at a time in the order of $\tau$. The algorithm is required to make an irrevocable decision on how to match (if at all) the newly arriving online node $u$. An \emph{adaptive priority} algorithm can change the permutation $\tau$ of the universe of data items after each arriving node $u$. This gives us two models of deterministic priority algorithms. There are two natural ways of introducing randomness in this model: (1) allow the algorithm to use randomness in selecting $\tau$, (2) allow the algorithm to use randomness for making decisions on how to match the newly arriving node $u$. This leads to 8 different possible priority models. 
In this paper, we shall be interested in an adaptive priority randomized algorithm called \textsc{MinRanking} (see Section~\ref{sec:priority}). 
As defined, \textsc{MinRanking} is a fully randomized priority algorithm, (i.e. in terms of both the ordering and decisions being randomized) but it can also be modified to be a randomized priority algorithm in the sense that 
its irrevocable decisions are randomized (as done in the \textsc{Ranking} algorithm) whereas the 
input sequence is 
chosen deterministically.

\section{Deterministic Multipass Online Algorithms}
\label{sec:multipass}
A simple greedy algorithm achieves approximation ratio $1/2$ in the online adversarial model since it computes a maximal matching. Moreover, no deterministic online algorithm can achieve an approximation guarantee better than $1/2$ in the online adversarial model. To break through the $1/2$ barrier either the input or algorithmic model needs to be changed. D\"urr \etal.  \cite{Durr2016} modify the algorithmic model to allow for a second pass over the input. They give a deterministic 2-pass algorithm, called \textsc{Category-Advice}, that achieves a $3/5$ approximation ratio. The \textsc{Category-Advice} algorithm belongs to the class of \textsc{Ranking}-based algorithms called category algorithms that were introduced in the work of D\"urr \etal.

A category algorithm considers a permutation $\sigma$ of the offline nodes. Instead of running \textsc{Ranking} directly with $\sigma$, a category function $c: V \rightarrow \mathbb{Z}\cup \{\pm \infty\}$ is computed first. The updated permutation $\sigma_c$ is the unique permutation satisfying the following defining property: for all $v_1, v_2 \in V$, we have $\sigma_c(v_1) < \sigma_c(v_2)$ if and only if $c(v_1) < c(v_2)$ or ($c(v_1) = c(v_2)$ and $\sigma(v_1) < \sigma(v_2)$). Then \textsc{Ranking} is performed with $\sigma_c$ as the permutation of the offline nodes. In other words, a category algorithm partitions the offline nodes into $|\Ima(c)|$ categories and specifies the ranking of the categories, the ranking within the category is induced by the initial permutation $\sigma$.

The \textsc{Category-Advice} algorithm starts with an arbitrary permutation $\sigma$, e.g., $\sigma$ could be induced by the names of nodes $V$. In the first pass, the algorithm runs \textsc{Ranking} with $\sigma$. Let $M$ be the matching obtained in the first pass. The category function $c: V \rightarrow [2]$ is then defined as follows: $c(v) = 1$ if $v \not\in M$ and $c(v) = 2$ otherwise. In the second pass, the \textsc{Category-Advice} algorithm runs \textsc{Ranking} with $\sigma_c$. The output of the second run of \textsc{Ranking} is declared as the output of the \textsc{Category-Advice} algorithm. In other words, in the second pass the algorithm gives preference to those vertices that were {\bf not} matched in the first pass. 
(We observe that the Besser-Poloczek graphs, as depicted in 
Figure \ref{figure:besser-poloczek}, show that the category function $c(v) = $ 
degree of $v$ will not yield an asymptotic approximation better than $1/2$.)  

In this section, we present a natural extension of the \textsc{Category-Advice} algorithm to multiple passes, called \textsc{$k$-pass Category-Advice} (see Algorithm~\ref{algo:ext-cat-adv}). Let $F_n$ denote the $n$th Fibonacci number, i.e., $F_1 = F_2 = 1$, and $F_n = F_{n-1} + F_{n-2}$ for $n \ge 3$. For each $k \ge 1$, we prove that the \textsc{$k$-pass Category-Advice} algorithm\footnote{A notable feature of this multi-pass algorithm is that after pass $i$, the algorithm can deterministically commit to matching a subset of size $\frac{F_{2i}}{F_{2i+1}}|M|$ where $M$ is a maximum matching. This will follow from
Lemma \ref{lem:ranking-mono}.} achieves the approximation ratio $F_{2k}/F_{2k+1}$. Moreover, we show that this is tight, i.e., there exists a family of bipartite graphs, one for each $k \ge 1$, such that the \textsc{$k$-pass Category-Advice} algorithm computes a matching that is $F_{2k}/F_{2k+1}$ times the size of the maximum matching. In particular, we show that the analysis of the $3/5$ approximation ratio in \cite{Durr2016} is tight. 
It immediately follows from our results that as $k$ goes to infinity, the approximation guarantee of the \textsc{$k$-pass Category-Advice} algorithm tends to the inverse of the golden ratio $2/(1+\sqrt{5}) \approx 0.618$. The convergence is extremely fast -- the approximation ratio guaranteed by the $5$-pass algorithm is already within $0.01\%$ of the inverse of the golden ratio. The following main theorem of this section follows immediately from Lemmas~\ref{lem:k-pass-main-pos} and \ref{lem:k-pass-main-neg}.

\begin{theorem}
The exact approximation ratio of the  \textsc{$k$-pass Category-Advice} algorithm is $F_{2k}/F_{2k+1}$, where $F_n$ is the $n$th Fibonacci number. Thus, the approximation ratio of the \textsc{$k$-pass Category-Advice} algorithms tends to the inverse of the golden ration $2/(1+\sqrt{5}) \approx 0.618$ as $k$ goes to infinity. This holds even when $k$ is allowed to depend on $n$ arbitrarily. For constant $k$, the algorithm can be implemented as a semi-streaming algorithm using $O(n)$ space.\footnote{To the best of our knowledge, published semi-streaming algorithms for bipartite matching use space ${\tilde \Omega}(n$) where the ${\tilde \Omega}$ hides polylogarithmic factors.}
\end{theorem}

\subsection{Positive Result}

The 2-pass algorithm in \cite{Durr2016} is called \textsc{Category-Advice}, so we refer to our generalization as \textsc{$k$-pass Category-Advice}. The pseudocode appears in Algorithm~\ref{algo:ext-cat-adv}. The algorithm is defined iteratively with each iteration corresponding to a new pass. The algorithm initializes $\sigma$ of the offline nodes to the identity permutation. 
The algorithm maintains a category function $c: V \rightarrow \mathbb{Z}\cup\{\pm \infty\}$. Initially, $c$ is set to $-\infty$ everywhere. In the $i$th pass, the algorithm runs \textsc{Ranking} on $\sigma_c$. Let $M_i$ be the resulting matching. For each $v \in V$, if $c(v) = -\infty$ and $v \in M_i$ then $c(v)$ is updated to $-i$. The algorithm uses the updated $c$ in the next pass. In words, $c$ records for each vertex $v$ the (negative of the) first pass, in which $v$ was matched. In the subsequent pass, the algorithm gives preference to the nodes that were unmatched, followed by nodes that were matched for the first time in the latest round, etc.

\begin{algorithm}[!h]
\caption{The \textsc{$k$-pass Category-Advice} algorithm.}\label{algo:ext-cat-adv}
\begin{algorithmic}
\Procedure{$k$-pass Category Advice}{}
\State{$\sigma \leftarrow $ the identity permutation of $V$}
\State{initialize array $c$ of size $V$ with $-\infty$}
\For{$i$ from $1$ to $k$}
\State{Perform the $i$th pass: $M_i \leftarrow $\textsc{Ranking}$(\sigma_c)$}
\For{$v \in V$}
\If{$c(v) = -\infty$ and $v \in M_i$}
\State $c(v) \leftarrow -i$
\EndIf
\EndFor
\EndFor
\Return{$M_k$}
\EndProcedure
\end{algorithmic}
\end{algorithm}

Before we prove the main result of this subsection, we recall several useful results regarding the analysis of the \textsc{Ranking} algorithm.

\begin{lemma}[Lemma 2 in Birnbaum and Mathieu~\cite{BirnbaumM2008}]
\label{lem:perfect-match-assump}
Let $G=(U,V,E)$ be a given bipartite graph. Let $\pi$ be a permutation of the online side, and $\sigma$ be a permutation of the offline side. Let $x \in U \cup V$, and set $H = G \setminus x$. Let $\pi_H$ and $\sigma_H$ be the permutations of $U_H$ and $V_H$ induced by $\pi$ and $\sigma$. If the matchings \textsc{Ranking}$(\pi_H, \sigma_H)$ and \textsc{Ranking}$(\pi, \sigma)$ are not identical, then they differ by a single alternating path starting at $x$.
\end{lemma}
Let $M$ be a maximum matching in $G$. Choose $x \not\in M$ and delete it. The above lemma implies that the size of the matching found by \textsc{Ranking} on $H$ is at most the size of the matching found by \textsc{Ranking} on $G$, while the size of a maximum matching hasn't changed. We can repeat this procedure until the resulting graph has a perfect matching $M$. This means that the worst-case approximation ratio of \textsc{Ranking} is achieved on bipartite graphs with a perfect matching. Since the \textsc{$k$-pass Category-Advice} algorithm is a \textsc{Ranking}-based algorithm, the same conclusion holds for \textsc{$k$-pass Category-Advice}. Hence, from now on we assume that the input graph $G=(U,V,E)$ is such that $|U|=|V|=n$ and $G$ has a perfect matching. Before proving the main result of this subsection, we need one more lemma regarding \textsc{Ranking}.

\begin{lemma}[Implicit in the proof of Lemma 4 in Birnbaum and Mathieu~\cite{BirnbaumM2008}]\label{lem:ranking-mono}
Let $\sigma_1$ be a permutation of the offline vertices and let $M_1 = $\textsc{Ranking}$(\sigma_1)$. Let $v$ be an offline vertex such that $v \not\in M_1$. Let $\sigma_2$ be a permutation of the offline vertices obtained from $\sigma_1$ by improving the rank of $v$. Let $M_2 = $\textsc{Ranking}$(\sigma_2)$. Then for every online node $u$ if $u \in M_1$ then $u \in M_2$ and $\sigma_2(M_2(u)) \le \sigma_2(M_1(u))$.
\end{lemma}

This lemma immediately implies the following:
\begin{lemma}\label{lem:k-pass-mono}
Let $M_k$ be the matching obtained by running the \textsc{$k$-pass Category-Advice} algorithm on some instance. Let $M_{k'}$ be the matching obtained by running the \textsc{$k'$-pass Category-Advice} algorithm on the same instance. If $k \le k'$ then $|M_k| \le |M_{k'}|$.
\end{lemma}
\begin{proof}
Simple induction.
\end{proof}

We are now ready to prove the main result of this subsection.

\begin{lemma}\label{lem:k-pass-main-pos}
The \textsc{$k$-pass Category-Advice} algorithm achieves approximation ratio $F_{2k}/F_{2k+1}$.
\end{lemma}

\begin{proof}[Proof by induction on $k$.]
Base case: $k = 1$. The \textsc{$1$-pass Category-Advice} algorithm is the simple deterministic greedy algorithm, which achieves a $1/2 = F_2/F_3$ approximation ratio.

Inductive step. Assume that the  \textsc{$k$-pass Category-Advice} algorithm achieves approximation ratio $F_{2k}/F_{2k+1}$. Let's consider the \textsc{$(k+1)$-pass Category-Advice} algorithm running on some bipartite input graph $G=(U \cup V, E)$, where $U$ are the online nodes. By Lemma~\ref{lem:perfect-match-assump}, we may assume without loss of generality that $|U|=|V|$ and $G$ has a perfect matching. Let $U_1 \subseteq U$ be the set of nodes matched in the first pass of the algorithm. Let $V_1 \subseteq V$ be the nodes that $U_1$ nodes are matched to. Define $U_2 = U \setminus U_1$ and $V_2 = V \setminus V_1$. Note that there are no edges between $U_2$ and $V_2$; otherwise some node of $U_2$ would have been matched to some node of $V_2$ in the first round. Let $M_i$ be the matching found by the \textsc{$(k+1)$-pass Category-Advice} algorithm in round $i$, where $i \in [k+1]$. Also define $M_{11} = M_{k+1} \cap U_1 \times V_1$, $M_{12} = M_{k+1} \cap U_1 \times V_2$, and $M_{21} = M_{k+1} \cap U_2 \times V_1$. We are interested in computing a bound on $|M_{k+1}| = |M_{11}| + |M_{12}| + |M_{21}|$. Figure~\ref{fig:pos-dem} is a graphical depiction of the variables described.

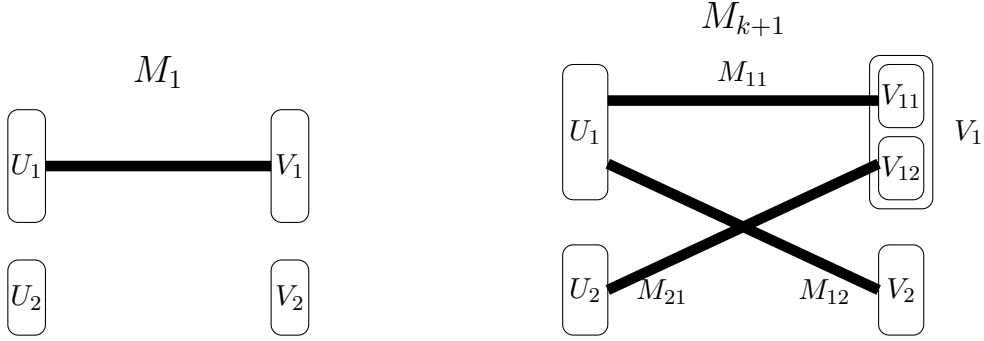
\begin{figure}[h]
\centering

\begin{tikzpicture}[scale=0.5]
\draw[rounded corners] (0,0) rectangle (-1,-3) node[pos=.5] {$U_1$}; 
\draw[rounded corners] (0,-4) rectangle (-1,-6) node[pos=.5] {$U_2$};
\draw[rounded corners] (7,0) rectangle (6,-3) node[pos=.5] {$V_1$}; 
\draw[rounded corners] (7,-4) rectangle (6,-6) node[pos=.5] {$V_2$};

\draw[line width=4] (0, -1.5) -- (6, -1.5);

\node at (3,1) {\Large $M_1$};

\end{tikzpicture} 
\hspace{3cm}
\begin{tikzpicture}[scale=0.6]
\draw[rounded corners] (0,0) rectangle (-1,-3) node[pos=.5] {$U_1$}; 
\draw[rounded corners] (0.0,-4) rectangle (-1,-6) node[pos=.5]{$U_2$} ;
\draw[rounded corners] (7.2,0.2) rectangle (5.8,-3.2); 
\draw[rounded corners] (7,-4) rectangle (6,-6) node[pos=.5] {$V_2$};

\draw[rounded corners] (7,0) rectangle (6,-1.4) node[pos=.5] {$V_{11}$};
\draw[rounded corners] (7,-1.6) rectangle (6,-3) node[pos=.5] {$V_{12}$};

\draw[line width=4] (0, -0.8) -- (6, -0.8) node[above,pos=.5]{$M_{11}$};
\draw[line width=4] (0, -2.2) -- (6, -5) node[below, pos=.8]{$M_{12}$};
\draw[line width=4] (0, -5) -- (6, -2.2) node[below,pos=.2]{$M_{21}$};

\node at (3, 1) {\Large $M_{k+1}$};
\node at (8, -1.5) {$V_1$};
\end{tikzpicture} 

\caption{$G$ is the input graph. On the left we show the matching constructed in the first pass. On the right we show the matching constructed in the $k+1$st pass.}\label{fig:pos-dem}
\end{figure}

Observe that after the first round, nodes in $U_1$ prefer nodes from $V_2$ to those from $V_1$. Moreover, nodes in $V_2$ are only connected to nodes in $U_1$ and there is a perfect matching between $V_2$ and a subset of $U_1$. Thus, the matching constructed between $U_1$ and $V_2$ in the next $k$ passes is the same as if we ran \textsc{$k$-pass Category-Advice} algorithm on the subgraph of $G$ induced by $U_1 \cup V_2$. This implies that $|M_{12}| \ge (F_{2k}/F_{2k+1}) |V_2| = (F_{2k}/F_{2k+1}) |U_2|$.

By Lemma~\ref{lem:ranking-mono}, in the $k+1$st pass, all nodes from $U_1$ that were not matched with $V_2$ will be matched with some nodes in $V_1$, i.e., $|U_1| = |M_{12}| + |M_{11}|$. Let $V_{11}$ be such a set, and let $V_{12} = V_1 \setminus V_{11}$. We would like to lower bound $|M_{21}|$. To that end we first lower bound the size of a maximum matching between $U_2$ and $V_{12}$. Since $U_2$ is only connected to $V_1$ and since we assume there is a perfect matching, a maximum matching between $U_2$ and $V_1$ is of size $|U_2|$. Thus, the size of a maximum matching between $U_2$ and $V_{12}$ is at least $|U_2|-|V_{11}|$. Also, observe that $|V_{11}| = |V_1|-|V_{12}|$ and $|V_{12}| = |M_{12}|$. Therefore, the size of a maximum matching between $U_2$ and $V_{12}$ is at least $|U_2|- (|V_1|-|M_{12}|) = |U_2| - |U_1| + |M_{12}|$ (note that $|U_1| = |V_1|$). Finally, observe that in the last round, the algorithm constructs a maximal matching between $U_{2}$ and $V_{12}$ guaranteeing that $|M_{21}| \ge (1/2) (|U_2|-|U_1|+|M_{12}|)$.

Putting it all together, we obtain
\begin{align*}
|M_{k+1}| &= |M_{11}| + |M_{12}| + |M_{21}| = |U_1| - |M_{12}| + |M_{12}| + |M_{21}| \\
          &\ge |U_1| + \frac{1}{2} \left( |U_2| - |U_1| + |M_{12}| \right) \ge \frac{1}{2} \left(|U_2| + |U_1| + \frac{F_{2k}}{F_{2k+1}} |U_2| \right) \\
          &= \frac{1}{2} \left( n + \frac{F_{2k}}{F_{2k+1}} (n-|M_1|) \right) = \left( \frac{1}{2} + \frac{F_{2k}}{F_{2k+1}} \right) n - \frac{F_{2k}}{2F_{2k+1}} |M_1|.\\
\end{align*}

By Lemma~\ref{lem:k-pass-mono} we also have that $|M_{k+1}| \ge |M_1|$. Thus, we derive
\[ |M_{k+1}| \ge \max \left\{ |M_1|, \left( \frac{1}{2} + \frac{F_{2k}}{F_{2k+1}} \right) n - \frac{F_{2k}}{2F_{2k+1}} |M_1| \right\}.\]
This implies that $|M_{k+1}| \ge \frac{1/2 + F_{2k}/F_{2k+1}}{1+F_{2k}/(2F_{2k+1})} n = \frac{F_{2k+1}+F_{2k}}{2F_{2k+1}+F_{2k}} n = \frac{F_{2(k+1)}}{F_{2(k+1)+1}} n$.
\end{proof}

\subsection{Negative Result}

In this subsection, we construct a family of bipartite graphs $(G_k)_{k=1}^\infty$ with the following properties:
\begin{enumerate}
\item $G_k$ has exactly $F_{2k+1}$ online nodes and $F_{2k+1}$ offline nodes.
\item $G_k$ has a perfect matching.
\item The \textsc{$k$-pass Category-Advice} algorithm finds a matching of size $F_{2k}$.
\item For all $k' > k$, the \textsc{$k'$-pass Category-Advice} algorithm finds a matching of size $F_{2k}+1$.
\end{enumerate}
\begin{remark}
The last property allows us to conclude that the approximation ratio of \textsc{$k$-pass Category-Advice} converges to the inverse of the golden ratio even when $k$ is allowed to depend on $n$ arbitrarily.
\end{remark}
This family of graphs shows that our analysis of the \textsc{$k$-pass Category-Advice} algorithm is tight. First, we describe the construction and then prove the above properties.

The construction of the $G_k$ is recursive. The base case is given by $G_1$, which is depicted in Figure~\ref{fig:basecase}. The online nodes are shown on the left whereas the offline nodes are on the right. The online nodes always arrive in the order shown on the diagram from top to bottom, and the initial permutation $\sigma$ of the offline nodes is also given by the top to bottom ordering of the offline nodes on the diagram.

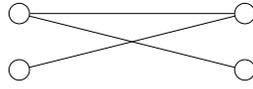
\begin{figure}[h]
\centering
\scalebox{0.75}{

\begin{tikzpicture}

\node[fill=white,circle,draw] (a) at (-2,1) {};
\node[fill=white,circle,draw] (b) at (-2,0) {};
\node[fill=white,circle,draw] (a2) at (2,1) {};
\node[fill=white,circle,draw] (b2) at (2,0) {};

\draw (a) -- (a2);
\draw (a) -- (b2);
\draw (b) -- (a2);
\end{tikzpicture}

}

\caption{$G_1$ is the graph used for the basis of the induction.}\label{fig:basecase}
\end{figure}

In the recursive step, we assume that $G_k$ has been constructed, and we show how to construct $G_{k+1}$. The online nodes $U$ of $G_{k+1}$ are partitioned into three disjoint sets $U = U_1 \cup U_2 \cup U_3$ such that $|U_1| =|U_3|= F_{2k+1}$ and $|U_2| = F_{2k}$. Similarly, the offline nodes $V$ of $G_{k+1}$ are partitioned into three disjoint sets $V = V_1 \cup V_2 \cup V_3$ such that $|V_1| = |V_3| = F_{2k+1}$ and $|V_2| = F_{2k}$. There is a copy of $G_k$ between $U_1$ and $V_3$. $U_2$ and $V_2$ are connected by parallel edges. There is a complete bipartite graph between $U_1$ and $V_1$, as well as between $U_2$ and $V_1$. Finally, $U_3$ is connected to $V_1$ by parallel edges. The construction is depicted in Figure~\ref{fig:recursive}.

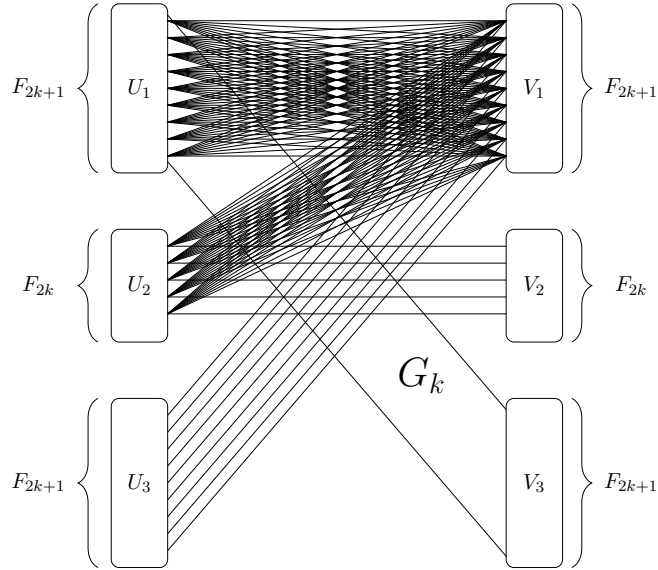
\begin{figure}[h]
\centering
\scalebox{0.75}{

\begin{tikzpicture} 
\draw[rounded corners] (0,0) rectangle (-1,-3) node[pos=.5] {$U_1$}; 
\draw[rounded corners] (0,-4) rectangle (-1,-6) node[pos=.5] {$U_2$};
\draw[rounded corners] (0, -7) rectangle (-1, -10) node[pos=.5] {$U_3$};
\draw[rounded corners] (7,0) rectangle (6,-3) node[pos=.5] {$V_1$}; 
\draw[rounded corners] (7,-4) rectangle (6,-6) node[pos=.5] {$V_2$};
\draw[rounded corners] (7, -7) rectangle (6, -10) node[pos=.5] {$V_3$};

\foreach \i in {1,...,5}
{
  \draw ($ (0,-4) + \i*(0,-0.3) $) -- ($ (6,-4) + \i*(0,-0.3) $);
}

\foreach \i in {1,...,9}
{
  \draw ($ (0,-7) + \i*(0,-0.3) $) -- ($ (6,0) + \i*(0,-0.3) $);
}

\foreach \i in {1,...,5}
{
  \foreach \j in {1,...,9}
  {
    \draw ($ (0,-4) + \i*(0,-0.3) $) -- ($ (6,0) + \j*(0,-0.3) $);
  }
}

\foreach \i in {1,...,9}
{
 \foreach \j in {1,..., 9}
 {
   \draw ($ (0,0) + \i*(0,-0.3) $) -- ($ (6,0) + \j*(0,-0.3) $);
 }
}

\draw (0, -0.2) -- (6, -7.2) node[pos=0.75,below,yshift=-20]{\huge $G_k$};
\draw (0, -2.8) -- (6, -9.8);

\draw [decorate,decoration={brace,amplitude=10},xshift=-1,yshift=0]
(-1.2,-3) -- (-1.2,0) node [black,midway,xshift=-30] {$F_{2k+1}$};

\draw [decorate,decoration={brace,amplitude=10},xshift=-1,yshift=0]
(-1.2,-6) -- (-1.2,-4) node [black,midway,xshift=-30] {$F_{2k}$};

\draw [decorate,decoration={brace,amplitude=10},xshift=-1,yshift=0]
(-1.2,-10) -- (-1.2,-7) node [black,midway,xshift=-30] {$F_{2k+1}$};

\draw [decorate,decoration={brace,amplitude=10},xshift=-1,yshift=0]
(7.2,0) -- (7.2,-3) node [black,midway,xshift=30] {$F_{2k+1}$};

\draw [decorate,decoration={brace,amplitude=10},xshift=-1,yshift=0]
(7.2,-4) -- (7.2,-6) node [black,midway,xshift=30] {$F_{2k}$};

\draw [decorate,decoration={brace,amplitude=10},xshift=-1,yshift=0]
(7.2,-7) -- (7.2,-10) node [black,midway,xshift=30] {$F_{2k+1}$};

\end{tikzpicture} 

}

\caption{$G_{k+1}$ depicts the inductive construction.}\label{fig:recursive}
\end{figure}

\begin{lemma}\label{lem:k-pass-main-neg}
Properties (1-4) mentioned at the beginning of this subsection hold for the $G_k$.
\end{lemma}
\begin{proof}[Proof by induction on $k$.]

Base case: $k = 1$. We have that $G_1$ has 2 online and 2 offline nodes and $F_3 = 2$. It is easy to see that $G_1$ has a perfect matching. The \textsc{1-pass Category-Advice} algorithm is simply the regular deterministic greedy algorithm, which clearly returns a matching of size $1$ on $G_1$, and $F_2 = 1$. Lastly, \textsc{$2$-pass Category-Advice} finds a matching of size $F_2+1 = 2$ and adding more passes after that does not change the matching found by the algorithm.

Inductive step. Assume that properties (1-4) hold for $G_k$. First, observe that the number of online vertices of $G_{k+1}$ is equal to the number of offline vertices and is equal to
\[ F_{2k+1} + F_{2k} + F_{2k+1} = F_{2k+2} + F_{2k+1} = F_{2k+3} = F_{2(k+1)+1}.\]
By inductive assumption, $G_k$ has a perfect matching. Therefore, $U_1$ vertices can be matched with $V_3$ via the perfect matching given by $G_k$. In addition, $U_2$ can be matched with $V_2$ by parallel edges, and $U_3$ can be matched with $V_1$ by parallel edges as well. Thus, $G_{k+1}$ has a perfect matching. Thus, we proved properties 1 and 2. Recall, that the order in which online vertices $U$ appear is the top-to-bottom order in Figure~\ref{fig:recursive}. To prove the 3rd property, observe that in the first pass, the \textsc{$(k+1)$-pass Category-Advice} algorithm matches $U_1$ and $V_1$ by parallel edges, $U_2$ with $V_2$ by parallel edges, and leaves $U_3$ and $V_3$ unmatched. Since $V_3$ is only connected to the nodes $U_1$, in the next $k$ passes the behavior of the algorithm between $U_1$ and $V_3$ nodes is exactly that of the \textsc{$k$-pass Category-Advice} algorithm. Therefore, by the inductive assumption, the algorithm is going to match exactly $F_{2k}$ nodes from $U_1$ with the nodes in $V_3$. The remaining $F_{2k+1}-F_{2k}$ nodes from $U_1$ will be matched to the nodes in $V_1$ (since those are the only neighbors of $U_1$ nodes besides the nodes from $V_3$). Also note that the nodes from $U_2$ in all passes behave the same way -- they prefer $V_1$ nodes to $V_2$ nodes. Thus, since $V_1$ will have $F_{2k}$ nodes unmatched after processing all nodes of $U_1$ in the last round, all of $U_2$ nodes will be matched to $V_1$ in the last round. This implies that after processing $U_1$ and $U_2$ in the last round, all of $V_1$ nodes are matched. Therefore, none of $U_3$ nodes can be matched. Thus, the matching found by the \textsc{$(k+1)$-pass Category-Advice} algorithm on $G_{k+1}$ is of size $|U_1|+|U_2| = F_{2k+1}+F_{2k} = F_{2k+2} = F_{2(k+1)}$. The last property is proved similarly to the 3rd. Let $k' > k$ be given and we analyze the behavior of \textsc{$k'$-pass Category-Advice} on $G_{k+1}$. After the first pass, by inductive assumption, the next $k'-1$ passes are going to produce a matching of size $F_{2k}+1$ between $U_1$ and $V_3$. This means that in the last pass, $F_{2k+1}-(F_{2k}+1)$ vertices from $U_1$ will be matched with $V_1$. This, in turn, implies that $F_{2k}$ vertices from $U_2$ will be matched with $V_1$, leaving $1$ more vertex from $V_1$ to be matched with $U_3$. This results in the overall matching of size $F_{2k+1}+F_{2k}+1 = F_{2k+2}+1$.
\end{proof}

\section{Analysis of \textsc{MinDegree} in the Known IID Model}
\label{sec:stochastic}
In this section we prove a tight approximation ratio for a natural greedy algorithm in the known IID input model---\textsc{MinDegree}.   
Recall from Section~\ref{sec:prelim} that every algorithm achieving approximation ratio $\rho$ in the known IID model can be converted into a greedy algorithm achieving approximation ratio $\ge \rho$. Thus, we can consider greedy algorithms without loss of generality. Moreover, greedy algorithms satisfying natural consistency conditions (see Definition~\ref{def:consistency}) achieve approximation ratio at least $1-1/e$ (see Theorem~\ref{thm:known_iid_greedy_positive}). Ties may occur in a greedy algorithm when an online node has more than one available neighbor. A greedy algorithm can be made more specialized if we specify a tie-breaking rule. How important is the tie-breaking for the performance of a greedy algorithm in the known IID model? Turns out that it's very important and a good tie breaking rule can improve the approximation ratio.

Algorithms beating $1-1/e$ ratio are known in this model (\cite{Bahmani2010,Manshadi2011,HaeuplerMZ2011,Jaillet2014,BrubachSSX2016}). Although these algorithms are not stated as greedy algorithms, we can apply the general conversion to turn them into greedy algorithms, albeit with unnatural tie-breaking conditions. The tie-breaking rules of these algorithms often require polynomial time preprocessing of the type graph, thus they are feasible from the theoretical point of view, but not from a practical point of view for large type graphs. This motivates the study of conceptually simple and practical tie-breaking rules for greedy algorithms in the known IID model. We contribute to this by studying a natural tie-breaking rule.

\begin{remark}
In fact, it is easy to see that there is an optimal tie breaking rule that can be computed by dynamic programming. This leads to an optimal online algorithm in the known IID model. Unfortunately, the size of the dynamic programming table is exponentially large, so computing this optimal tie-breaking rule using such a 
dynamic program is not computationally feasible. Currently  no one knows how to analyze an optimal tie-breaking rule, but certain potentially sub-optimal and computationally feasible tie-breaking rules have been analyzed as discussed above.
\end{remark}

Our algorithm is \textsc{MinDegree}. The motivation for studying this algorithm is as follows. It is easy to see that in the adversarial setting we can take any greedy algorithm, modify the tie-breaking rule to always give more preference to the offline nodes of degree $1$, and this will not decrease the approximation ratio of the algorithm. Generalizing this, we may come to the conclusion that online vertices should give preference to the offline vertices of smaller degrees. The problem is that in the adversarial setting we do not know the degrees of the offline nodes a priori; however, in the known IID setting we can estimate the degrees of the offline vertices from the type graph. This is precisely what \textsc{MinDegree} formalizes. The algorithm is given a type graph $G=(U,V,E)$ as input. It keeps track of a set $S$ of currently ``active'', i.e., not yet matched, offline nodes. When a new node $\widetilde{u}$ arrives, it is matched to its active neighbor of minimum degree in the type graph. The pseudocode of \textsc{MinDegree} is presented in  Algorithm~\ref{alg:mindegree}.

\begin{algorithm}[!h]
\caption{The \textsc{MinDegree} algorithm.}
\label{alg:mindegree}
\begin{algorithmic}
\Procedure{MinDegree}{$G=(U,V,E)$}
\State{Let $S = V$ denote the set of active offline nodes}
\Repeat
\State{Let $\widehat{u}$ denote the newly arriving online node}
\State{Let $N(\widehat{u})= \{ v \in S \mid (u,v) \in E\}$}
\If{$N(\widehat{u}) \neq \emptyset$}
\State{Let $v = \argmin_{v \in N(\widehat{u})} \deg(v)$}
\State{Match $\widehat{u}$ with $v$}
\State{Remove $v$ from $S$}
\EndIf
\Until{all online nodes have been processed}
\EndProcedure
\end{algorithmic}
\end{algorithm}

\begin{remark}
Our algorithm does not fully break ties, i.e., \textsc{MinDegree} takes \emph{some} neighbor of currently minimum degree. In practice, it means that ties are broken in some arbitrary way, e.g., by names of vertices. In our theoretical analysis, this means that the adversary is more powerful, as it can decide how the algorithm breaks these ties.
\end{remark}

\textsc{MinDegree} 
is a conceptually simple and promising algorithm in the known IID setting. Indeed, a version of \textsc{MinDegree} has been extensively studied in the offline setting (see Besser and Poloczek \cite{BesserP17} and references therein). Unfortunately, in spite of having excellent empirical performance as well as excellent performance under various random input models, it has a bad worst-case approximation ratio of $1/2$ in the offline adversarial setting~\cite{BesserP17}. As far as we know, this algorithm has not been analyzed in the known IID model. We obtain a result that, in spirit, is similar to the offline result, but the proof is different. Namely, we show that \textsc{MinDegree} cannot achieve an approximation ratio better than $1-1/e$, which is guaranteed by any consistent greedy algorithm in the known IID model. 

\begin{theorem}
\[ \AR(\textsc{MinDegree}) = 1-\frac{1}{e}\]
The negative result holds no matter which rule is used to break (remaining) ties in \textsc{MinDegree}.

\end{theorem}

First observe that $\AR(\textsc{MinDegree}) \ge 1-1/e$ immediately follows from Theorem~\ref{thm:known_iid_greedy_positive}, since it is easy to see that \textsc{MinDegree} satisfies the consistency conditions. Thus, the rest of this section is dedicated to proving $\AR(\textsc{MinDegree}) \le 1-1/e$. Namely, we prove the following.

\begin{theorem}[\textsc{MinDegree} negative result]\label{mindegree_main_thm}
There is a family of type graphs witnessing
\[ \AR(\textsc{MinDegree}) \le 1-1/e.\]
This result holds no matter which rule is used to break (remaining) ties in \textsc{MinDegree}.
\end{theorem}

\begin{figure}[h]
\centering

\begin{tikzpicture}[scale=0.5]

\def \L {2}
\def \N {4}
\def \a {8}
\def \offset {2}

\foreach \s in {1,...,\N}
{
	\foreach \t in {1,...,\L}
	{
		\pgfmathtruncatemacro\f{(\s-1)*\L+\t}
		\pgfmathsetmacro\r{\f+1.5*\s}
  		\node[fill=white,circle,draw,scale=0.5] (z_\f) at ({\offset+0},\r) {};
	}
}

\foreach \s in {1,...,\N}
{
	\foreach \t in {1,...,\L}
	{
		\pgfmathtruncatemacro\f{(\s-1)*\L+\t}
		\pgfmathsetmacro\r{\f+1.5*\s}
  		\node[fill=white,circle,draw,scale=0.5] (yt_\f) at ({\offset-7},\r) {};
	}
}

\begin{pgfonlayer}{background}

\foreach \s in {1,...,\N}
{
	\foreach \t in {\s,...,\N}
	{
		\foreach \q in {1,...,\L}
		{
			\foreach \p in {1,...,\L}
			{
			 	 \pgfmathsetmacro\resl{(\s-1)*\L+\q}
				 \pgfmathsetmacro\resr{(\t-1)*\L+\p}
				 \draw (yt_\resr) -- (z_\resl);
			}
		}
	}
}
\end{pgfonlayer}

\node[label] at ($ (yt_1)+(0,-1) $) {$ U $};
\node[label] at ($ (z_1)+(0,-1) $) {$ V $};
\foreach \s in {1,...,\N}
{
	\pgfmathsetmacro\res{\L*(\N-\s+1)}
	\node[label] at ($ (yt_\res)+(-1.4,-0.5) $){$U_\s$};
	\node[label] at ($ (z_\res)+(1,-0.5) $){$V_\s$};
}
\end{tikzpicture}

\caption{$G_{2,4}$.}\label{fig:gm24}
\end{figure}

To prove the above theorem we use the family of type graphs due to Goel and Mehta~\cite{GoelM2008} as a building block. Each graph in that family is indexed by two integers $L$ and $N$. Let $G_{L,N}(U,V,E)$ denote the $(L,N)$th graph in the family. The offline side is partitioned into $N$ blocks of size $L$: $V = \dot\bigcup_{i=1}^N V_i$. The online side is partitioned into $N$ blocks of size $L$, i.e., $U = \dot\bigcup_{i=1}^N U_i$. The edges are as follows: there is a complete biclique between $V_i$ and $U_j$ if and only if $i \ge j$. An example of this graph with $L=2$ and $N=4$ is depicted in Figure~\ref{fig:gm24}.

\begin{theorem}[Goel and Mehta~\cite{GoelM2008}]
\label{thm:simple_greedy_lb}
There is a tie-breaking criterion such that a simple greedy algorithm finds a matching of size $\le LN(1-1/e) + o(LN)$ on $G_{L,N}$ with high probability.
\end{theorem}

In fact, the implicit tie-breaking criterion in the work of Goel and Mehta on $G_{L,N}$ is equivalent to breaking ties by maximum degree, i.e., online nodes from $U_j$ are matched with a neighboring node from $V_i$ with largest possible $i$. Thus, the online nodes try to fill in the spaces on the $V$ side from bottom up to prevent future online nodes from having neighbors. The high level idea behind our construction is to use $K$ independent copies of $G_{L,N}$ and $N$ shared gadgets. The gadgets are used to make the offline degrees the same on each copy of $G_{L,N}$ without affecting the size of the matching found by the algorithm on a particular copy of $G_{L,N}$. Then we can invoke an adversarial argument to say that \textsc{MinDegree} on the augmented $G_{L,N}$ behaves exactly like the Goel and Mehta greedy 
algorithm on the original $G_{L,N}$. As the number of copies $K$ goes to infinity, the contribution of the shared gadgets to the size of the overall matching goes to zero. Thus, we get the approximation ratio of $1-1/e$ for \textsc{MinDegree}.

We denote our type graph by $G_{L,N,K}$. Let $G_{L,N}^i$ denote the $i$th copy of $G_{L,N}$ in our type graph for $i \in [K]$. Gadget $j$, denoted by $H_{j,L}$, simply consists of a biclique with the offline side of size $L + o(L)$ and the online side of size $L$. We let $H_L$ denote $\dot\bigcup_{j=1}^N H_{j,L}$. For each $i$, we connect the offline side of $k$th block of $G_{L,N}^i$ via a biclique with the online side of $H_{j,L}$ for $k \le j$. Figure~\ref{fig:our24} shows an example, where the offline side has the same number of nodes as the online side in each gadget for simplicity. In Figure~\ref{fig:our24}, black nodes indicate the online side and white nodes indicate the offline side.

\begin{figure}[h]
\centering

\begin{tikzpicture}[scale=0.5]

\def \L {2}
\def \N {4}
\def \a {8}
\def \offset {2}

\foreach \s in {1,...,\N}
{
	\foreach \t in {1,...,\L}
	{
		\pgfmathtruncatemacro\f{(\s-1)*\L+\t}
		\pgfmathsetmacro\r{\f+1.5*\s}
  		\node[fill=white,circle,draw,scale=0.5] (hz_\f) at ({\offset+10},\r) {};
	}
}

\foreach \s in {1,...,\N}
{
	\foreach \t in {1,...,\L}
	{
		\pgfmathtruncatemacro\f{(\s-1)*\L+\t}
		\pgfmathsetmacro\r{\f+1.5*\s}
  		\node[fill=black,circle,draw,scale=0.5] (hy_\f) at ({\offset+8},\r) {};
	}
}

\foreach \s in {1,...,\N}
{
	\foreach \t in {1,...,\L}
	{
		\pgfmathtruncatemacro\f{(\s-1)*\L+\t}
		\pgfmathsetmacro\r{\f+1.5*\s}
  		\node[fill=white,circle,draw,scale=0.5] (z_\f) at ({\offset+5},\r) {};
	}
}
\foreach \s in {1,...,\N}
{
	\foreach \t in {1,...,\L}
	{
		\pgfmathtruncatemacro\f{(\s-1)*\L+\t}
		\pgfmathsetmacro\r{\f+1.5*\s}
  		\node[fill=black,circle,draw,scale=0.5] (y_\f) at ({\offset+0},\r) {};
	}
}
\begin{pgfonlayer}{background}

\foreach \s in {1,...,\N}
{
	\foreach \t in {\s,...,\N}
	{
		\foreach \q in {1,...,\L}
		{
			\foreach \p in {1,...,\L}
			{
			 	 \pgfmathsetmacro\resl{(\s-1)*\L+\q}
				 \pgfmathsetmacro\resr{(\t-1)*\L+\p}
				 \pgfmathsetmacro\resh{(\t-1)*\L+\q}
				 \draw (y_\resr) -- (z_\resl);
				 \draw (hy_\resr) -- (hz_\resh);
				 \draw (hy_\resl) -- (z_\resr);
			}
		}
	}
}
\end{pgfonlayer}

\draw[dashed,rounded corners] ($(y_1)+(-.5,-.5)$) rectangle ($(z_\a)+(.5,+.5)$);
\draw[dashed,rounded corners] ($(hz_1)+(.5,-.5)$) rectangle ($(hy_\a)+(-.5,+.5)$);

\node[label] at ($ (z_1)+(-2,-1) $) {$\text{Copy of } G_{L,N} $};
\node[label] at ($ (hz_1)+(-1,-1) $) {$\text{Gadgets } H_L$};
\foreach \s in {1,...,\N}
{
	\pgfmathsetmacro\res{\L*(\N-\s+1)}
	\node[label] at ($ (hz_\res)+(+1.5,-0.5) $){$H_{\s,L}$};
}

\end{tikzpicture}

\caption{$L = 2, N= 4$.}\label{fig:our24}
\end{figure}

\begin{proof}[Proof of Theorem~\ref{mindegree_main_thm}]

Since each offline node in each copy $G_{L,N}^i$ has the same degree, we may assume that \textsc{MinDegree} breaks ties by matching an online node from block $j$ with an offline node of block $k$ for the largest possible $k$. This is the same behavior as in Goel and Mehta~\cite{GoelM2008} analysis. 
Assume that $L, N,$ and $K$ are large enough. By the Chernoff bound, with high probability the number of online nodes generated from $H_{j,L}$ is bounded by $L + o(L)$. Thus, by the choice of our gadget, with high probability all nodes from $H_{j,L}$ can be matched to the offline nodes within $H_{j,L}$. Therefore, the online nodes generated from gadgets do not have affect on any of the independent copies of $G_{L,N}$. Moreover, the expected number of nodes generated from $G_{L,N}^i$ is $LN$. Invoking the Chernoff bound again, we see that with high probability for each $i$ we generate at most $LN + o(LN)$ nodes from $G_{L,N}^i$. Since independent copies $G_{L,N}^i$ do not share any edges except with the gadget, the matchings constructed by our algorithm on independent copies are independent. Since the behavior of our algorithm on each copy of $G_{L,N}$ is identical to the behavior of a simple greedy algorithm, our algorithm for each $G_{L,N}^i$ constructs a matching of expected size $LN(1-1/e)+o(LN)$. Thus, overall our algorithm finds a matching of expected size $LNK(1-1/e) + o(LNK)$. Moreover, it is easy to see that a matching of size at least $LNK - o(LNK)$ is possible. The claimed approximation ratio follows.
\end{proof}

\begin{remark}
The claim of Theorem~\ref{mindegree_main_thm} can be strengthened by taking away the power of  the adversary in breaking ties. We can force the unfavorable tie-breaking used in the above proof by introducing an additional biclique gadget with $N + o(N)$ offline nodes and $N$ online nodes. We modify the degrees in $G_{L,N}^i$ in a natural way: the offline nodes from block $j$ of $G_{L,N}^i$ get connected to $N-j$ nodes of the online side of the new gadget. This forces the \textsc{MinDegree} algorithm to process blocks $G_{L,N}^i$ from ``bottom to top'' (see Figure~\ref{fig:gm24}), which is exactly the wrong order that results in $LNK(1-1/e)+o(LNK)$ matching. It is easy to see that the only other affect this additional gadget introduces is to increase the size of the matching found by \textsc{MinDegree}, as well as $\OPT$, by a lower order term $N+o(N) = o(LNK)$.
\end{remark}

\section{A Hybrid Algorithm in the Priority Model}
\label{sec:priority}
We propose a conceptually simple greedy algorithm for bipartite matching. Our algorithm is a natural hybrid between two  well-studied greedy algorithms for bipartite matching---\textsc{Ranking} (see Section~\ref{sec:prelim}) and \textsc{MinGreedy} (due to Tinhofer~\cite{Tinhofer84}). \textsc{MinGreedy} is an offline algorithm that selects a random vertex of minimum degree in the input graph and matches it with a random neighbor, removes the matched vertices, and proceeds. Despite excellent empirical performance and excellent performance under certain random input models, the algorithm achieves approximation ratio $1/2$ in the adversarial input setting (see \cite{BesserP17} and references therein). Algorithm~\ref{algorithm:mingreedy} shows the pseudocode of a natural adaptation of \textsc{MinGreedy} to bipartite graphs $G=(U,V,E)$. The algorithm picks a random node of minimum degree from $U$ and matches it to a random neighbor from $V$. Observe that this algorithm can be realized as a 
fully randomized priority algorithm, where the ordering of the input items is by increasing degree with ties broken randomly.

\begin{algorithm}
\caption{The \textsc{MinGreedy} algorithm.}\label{algorithm:mingreedy}
\begin{algorithmic}
\Procedure{MinGreedy}{$G=(U,V,E)$}
\Repeat
\State{Let $d = \min \{\deg(i) \mid i \in U \}$ and $S = \{i \in U \mid \deg(i) = d\}$}
\State{Pick $i \in S$ uniformly at random}
\State{Let $N(i)$ be the set of neighbors of $i$}
\If{$N(i) = \emptyset$}
\State{$i$ remains unmatched}
\State{Delete $i$ from $G$}
\Else
\State{Match $i$ with $j$ chosen from $N(i)$ uniformly at random}
\State{Delete $i$ and $j$ from $G$}
\EndIf
\Until{$U=\emptyset$}
\EndProcedure
\end{algorithmic}
\end{algorithm}

Karp, Vazirani, Vazirani~\cite{KarpVV90} exhibited a family of graphs, on which \textsc{Ranking} achieves its worst approximation ratio $1-1/e$. The biadjacency matrix of the $n$th graph in this family is an $n \times n$ upper triangular matrix with all $1$s at and above the diagonal. Interestingly, \textsc{MinGreedy} finds a perfect matching on these graphs. Thus, it is natural to consider the performance of an algorithm that combines \textsc{Ranking} and \textsc{MinGreedy}. This is our proposed \textsc{MinRanking} algorithm (see Algorithm~\ref{algorithm:minranking}). In \textsc{MinRanking}, a random permutation $\pi$ of vertices $V$ is initially sampled. Then, nodes in $U$ are processed in the increasing order of their current degrees with ties broken randomly. When a node $u$ is processed, it is matched with its neighbor appearing earliest in the ordering  $\pi$.

\begin{algorithm}
\caption{The \textsc{MinRanking} algorithm.}\label{algorithm:minranking}
\begin{algorithmic}
\Procedure{MinRanking}{$G=(U,V,E)$}
\State{Pick a permutation $\pi : V\rightarrow V$ uniformly at random}
\Repeat
\State{Let $d = \min \{\curdeg(i) \mid i \in U \}$ and $S = \{i \in U \mid \curdeg(i) = d\}$}
\State{Pick $i \in S$ uniformly at random}
\State{Let $N(i)$ be the set of neighbors of $i$}
\If{$N(i) = \emptyset$}
\State{$i$ remains unmatched}
\State{Delete $i$ from $G$}
\Else
\State{Match $i$ with $j = \argmin_k \{\pi(k) \mid k \in N(i)\}$}
\State{Delete $i$ and $j$ from $G$}
\State{Update $\curdeg$}
\EndIf
\Until{$U=\emptyset$}
\EndProcedure
\end{algorithmic}
\end{algorithm}

\textsc{MinrRanking} modifies \textsc{MinGreedy} in the same way 
that the online 
\textsc{Ranking} algorithm modifies the seemingly more natural online randomized algorithm that simply matches an online vertex to an available neighbor uniformly at random which surprisingly was shown to be (asymptotically) a 
$1/2$ approximation. So it is hopeful that \textsc{MinRanking} can improve upon \textsc{MinGreedy}.
Like \textsc{MinGreedy}, our algorithm can be implemented and analyzed in the 
fully randomized adaptive priority model \cite{BorodinNR2003}. Since our algorithm is a generalization of \textsc{Ranking}, its asymptotic approximation ratio is at least $1-1/e \approx 0.6321$. The main result of this section is that the asymptotic approximation ratio of this algorithm is at most $1/2 + 1/(2e)\approx 0.6839$, as witnessed by the family of graphs due to Besser and Poloczek \cite{BesserP17}. 

\begin{theorem}\label{thm:priority_intro_thm}
\[ 1-\frac{1}{e} \le \AR(\textsc{MinRanking}) \le \frac{1}{2} + \frac{1}{2e}.\]
\end{theorem}

This negative result shows that \textsc{MinRanking} falls short of the known bound for \textsc{Ranking} in the ROM model, where it achieves approximation ratio $0.696$~\cite{Mahdian2011}. From our result it follows that a deterministic ordering of the online nodes by non-decreasing degree (breaking ties by the given adversarial ordering of those nodes) will also fall short. That is (similar to the result in \cite{PenaB16} for deterministic decisions), a naive randomized ordering  can be better than a seemingly informed deterministic ordering.

\subsection{The Besser-Poloczek Graph Construction}
In this section we describe the family of bipartite graphs constructed by Besser and Poloczek in~\cite{BesserP17}. Each graph in the family is indexed by an integer parameter $b$. Let $G_b(L,R,E)$ denote the $b$th graph in the family. A note about notation: we use $L$ and $R$ instead of $U$ and $V$ in the definition of $G_b$. This is because we want to reserve $U$ and $V$  for a different graph that will play the central role in the analysis of \textsc{MinRanking}. Continuing with the definition of $G_b(L,R,E)$, each of the two sides $L$ and $R$ is partitioned into three sets:
\begin{itemize}
\item $L = S_{1,L} \dot\cup S_{2,L} \dot\cup S_{3,L}$
\item $R = S_{1,R} \dot\cup S_{2,R} \dot\cup S_{3,R}$
\end{itemize}
Blocks $S_{2,L}$ and $S_{2,R}$ are further partitioned into $b$ sets each:
\begin{itemize}
\item $S_{2,L} = \dot\bigcup_{i=1}^b S_{2,L}^{(i)}$
\item $S_{2,R} = \dot\bigcup_{i=1}^b S_{2,R}^{(i)}$
\end{itemize}
For convenience, we give numeric names to each vertex in $L$ and $R$:
\begin{itemize}
\item $S_{1,W} = \{1_{W}, 2_{W}, \ldots, b^2_{W}\}$ where $W\in \{L,R\}$
\item $S_{2,W}^{(i)} = \{(b^2+(i-1)b+1)_W,\ldots, (b^2+i b)_W\}$ where $W\in \{L,R\}$
\item $S_{3,W}= \{(2b^2+1)_W,\ldots, (2b^2+2b)_W\}$  where $W\in \{L,R\}$
\end{itemize}
Thus, we have $|L|=|R|=2b^2+2b$. Next, we describe the edges of $G_b$:
\begin{itemize}
\item There is a biclique between $S_{3,L}$ and $S_{1,R}$ and between $S_{3,R}$ and $S_{1,L}$.
\item $S_{1,R}$ and $S_{2,L}$ are connected by ``parallel'' edges, i.e., $\{i_R,(b^2+i)_L\} \in E$ for all $i \in [b^2]$.
\item $S_{1,L}$ and $S_{2,R}$ are connected by ``parallel'' edges, i.e., $\{i_L,(b^2+i)_R\} \in E$ for all $i \in [b^2]$.
\item $S_{3,L}$ and $S_{3,R}$ are connected by ``parallel'' edges, i.e., $\{i_R,i_L\} \in E$ for all $i \in \{2b^2+1,\ldots, 2b^2+2b\}$.
\item There is a biclique between $S_{2,L}^{(i)}$ and $S_{2,R}^{(i)}$ for each $i \in [b]$.
\end{itemize}

Figure~\ref{figure:besser-poloczek} depicts the Besser-Poloczek graph with $b=3$. We note that $G_b$ has a perfect matching.

\begin{claim}
$G_b$ has a perfect matching, i.e., a matching of size $2b^2+2b$.
\end{claim}
\begin{proof}
Using the parallel edges, match the vertices in $S_{1,R}$ with the vertices in $S_{2,L}$, the vertices in $S_{2,R}$ with the vertices in $S_{1,L}$, and the vertices in $S_{3,L}$ with the vertices in $S_{3,R}$.
\end{proof}

\begin{figure}[h]
\centering
\scalebox{0.75}{

\begin{tikzpicture}

\def \b {3}
\def \bb {6}
\def \a {9}
\def \offset {2}

\foreach \s in {1,...,\bb}
{
  \node (invl_\s) at (0,{\s+1.5}) {};
  \node (invr_\s) at ({2*\offset+11},{\s+1.5}) {};
}

\foreach \s in {1,...,\bb}
{
  \node[fill=black,circle,draw] (s3l_\s) at ({\offset+0},{\s+1.5}) {};
}

\foreach \s in {1,...,\a}
{
  \node[fill=white,circle,draw] (s1l_\s) at ({\offset+2},\s) {};
}

\foreach \s in {1,...,\a}
{
  \node[fill=black,circle,draw] (s2l_\s) at ({\offset+4},\s) {};
}

\foreach \s in {1,...,\a}
{
  \node[fill=white,circle,draw] (s2r_\s) at ({\offset+7},\s) {};
}

\foreach \s in {1,...,\a}
{
  \node[fill=black,circle,draw] (s1r_\s) at ({\offset+9},\s) {};
}

\foreach \s in {1,...,\bb}
{
  \node[fill=white,circle,draw] (s3r_\s) at ({\offset+11},{\s+1.5}) {};
}

\begin{pgfonlayer}{background}
\foreach \s in {1,...,\a}
{
  \draw (s2l_\s) -- (s1l_\s);
  \draw (s2r_\s) -- (s1r_\s);
}

\foreach \s in {1,...,\bb}
{
  \draw (s3l_\s) -- (invl_\s);
  \draw (s3r_\s) -- (invr_\s);
  \foreach \t in {1,...,\a}
  {
    \draw (s3l_\s) -- (s1l_\t);
    \draw (s3r_\s) -- (s1r_\t);
  }
}

\foreach \s in {1,...,\b}
{
	\foreach \t in {1,...,\b}
	{
		\foreach \q in {1,...,\b}
		{
			 \pgfmathsetmacro\resl{(\s-1)*\b+\t}
			 \pgfmathsetmacro\resr{(\s-1)*\b+\q}
			 \draw (s2r_\resr) -- (s2l_\resl);
		}
	}
}
\end{pgfonlayer}

\draw[dashed,rounded corners] ($(s3l_1)+(-.5,-.5)$) rectangle ($(s3l_\bb)+(+.5,+.5)$);
\draw[dashed,rounded corners] ($(s2l_1)+(-.5,-.5)$) rectangle ($(s2l_\a)+(+.5,+.5)$);
\draw[dashed,rounded corners] ($(s1l_1)+(-.5,-.5)$) rectangle ($(s1l_\a)+(+.5,+.5)$);
\draw[dashed,rounded corners] ($(s3r_1)+(-.5,-.5)$) rectangle ($(s3r_\bb)+(+.5,+.5)$);
\draw[dashed,rounded corners] ($(s2r_1)+(-.5,-.5)$) rectangle ($(s2r_\a)+(+.5,+.5)$);
\draw[dashed,rounded corners] ($(s1r_1)+(-.5,-.5)$) rectangle ($(s1r_\a)+(+.5,+.5)$);

\node[label] at ($ (s3l_1)+(0,-1) $) {$ S_{3,L} $};
\node[label] at ($ (s2l_1)+(0,-1) $) {$ S_{2,L} $};
\node[label] at ($ (s1l_1)+(0,-1) $) {$ S_{1,R} $};
\node[label] at ($ (s3r_1)+(0,-1) $) {$ S_{3,R} $};
\node[label] at ($ (s2r_1)+(0,-1) $) {$ S_{2,R} $};
\node[label] at ($ (s1r_1)+(0,-1) $) {$ S_{1,L} $};
\end{tikzpicture}

}

\caption{The Besser-Poloczek graph with parameter $b=3$. The edges on the left wrap around on the right connecting the nodes in $S_{3,L}$ to the nodes in $S_{3,R}$.}\label{figure:besser-poloczek}
\end{figure}
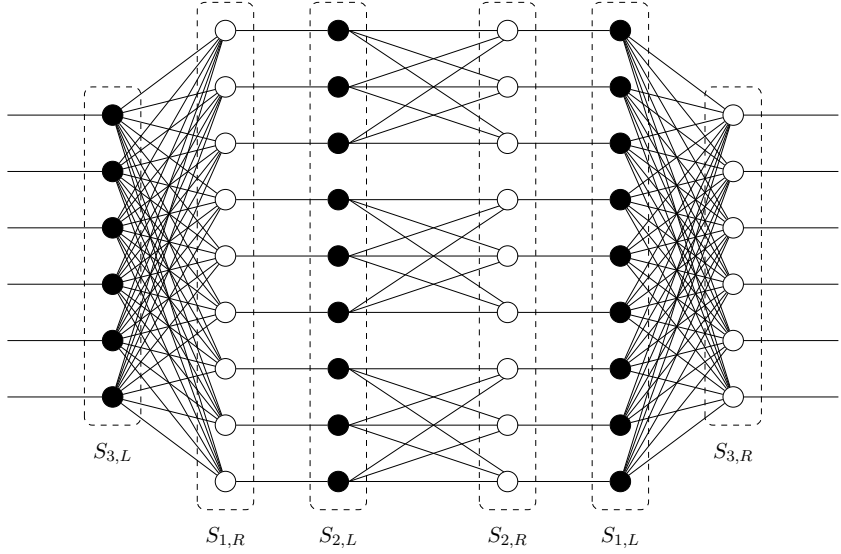

Besser and Poloczek~\cite{BesserP17} show that the graphs $G_b$ are asymptotically hardest for \textsc{MinGreedy}.

\begin{theorem}[Besser,Poloczek~\cite{BesserP17}]
The expected size of a matching constructed by \textsc{MinGreedy} on $G_b$ is $b^2+o(b^2)$. Thus, the worst case asymptotic approximation ratio of \textsc{MinGreedy} is $1/2$.
\end{theorem}

In the next section we shall see that \textsc{MinRanking} finds a significantly larger matching on $G_b$.

\subsection{Analysis of \textsc{MinRanking} on the Besser-Poloczek Graphs}

This section is devoted to proving the following theorem, which immediately implies Theorem~\ref{thm:priority_intro_thm}.

\begin{theorem}
\label{theorem:main_theorem}
\textsc{MinRanking} running on the family of graphs $\{G_b\}$ has the asymptotic approximation ratio exactly $\frac{1}{2} + \frac{1}{2e}$.
\end{theorem}

We shall prove the above theorem in two steps. In the first step, we get tight bounds on the expected number of edges of relevant type in a matching constructed by \textsc{MinRanking} running on certain induced subgraphs of the Besser-Poloczek graphs. In the second step, we show how to combine the results for subgraphs and derive the expected size of a matching of the entire Besser-Poloczek graph.

We begin by analyzing the expected performance of \textsc{MinRanking} on a particular bipartite graph $H_{n,k}(U,V,W)$ parameterized by two integers $n$ and $k$ where $k \le n$. The relevance of this analysis to Theorem~\ref{theorem:main_theorem} will be evident later. We use a different capital letter for these graphs to distinguish them from the Besser-Poloczek graphs. The left side consists of $n$ vertices $U = \{1_U, \ldots, n_U\}$. The right side is partitioned into two blocks $V = V_1 \dot\cup V_2$ such that the first block contains $k$ vertices $V_1 = \{1_{V_1}, \ldots, k_{V_1}\}$ and the second block contains $n$ vertices $V_2=\{1_{V_2},\ldots, n_{V_2}\}$. The edges consist of biclique edges between $U$ and $V_1$ and parallel edges between $U$ and $V_2$, i.e., we have $\{i_U,i_{V_2}\} \in W$ for $i \in [n]$. Figure~\ref{figure:hnk} depicts this bipartite graph for $n = 5$ and $k = 3$.

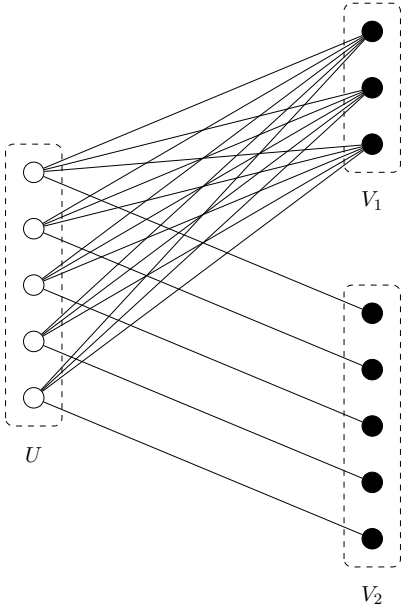
\begin{figure}[h]
\centering
\scalebox{0.75}{

\begin{tikzpicture}

\def \n {5}
\def \k {3}

\foreach \s in {1,...,\n}
{
  \node[fill=white,circle,draw] (u_\s) at (0,{\s+2.5}) {};
}

\foreach \s in {1,...,\n}
{
  \node[fill=black,circle,draw] (v2_\s) at (6,{\s}) {};
}

\foreach \s in {1,...,\k}
{
  \node[fill=black,circle,draw] (v1_\s) at (6,{\s+\k+4}) {};
}

\begin{pgfonlayer}{background}
\foreach \s in {1,...,\n}
{
  \draw (u_\s) -- (v2_\s);
}

\foreach \s in {1,...,\n}
{
  \foreach \t in {1,...,\k}
  {
    \draw (u_\s) -- (v1_\t);
  }
}
\end{pgfonlayer}

\draw[dashed,rounded corners] ($(u_1)+(-.5,-.5)$) rectangle ($(u_\n)+(+.5,+.5)$);
\draw[dashed,rounded corners] ($(v1_1)+(-.5,-.5)$) rectangle ($(v1_\k)+(+.5,+.5)$);
\draw[dashed,rounded corners] ($(v2_1)+(-.5,-.5)$) rectangle ($(v2_\n)+(+.5,+.5)$);

\node[label] at ($ (u_1)+(0,-1) $) {$ U $};
\node[label] at ($ (v1_1)+(0,-1) $) {$ V_{1} $};
\node[label] at ($ (v2_1)+(0,-1) $) {$ V_{2} $};
\end{tikzpicture}

}

\caption{The $H_{n,k}$ graph with $n = 5$ and $k = 3$.}\label{figure:hnk}
\end{figure}

Consider \textsc{MinRanking} running on $H_{n,k}$. For reasons that will become clear later, we are interested in obtaining sharp bounds on the expected number of $(U,V_2)$ edges in a matching output by $\textsc{MinRanking}$. Let $t \in \{0,1,\ldots,n\}$ denote the time, i.e., the iteration number of the main loop in \textsc{MinRanking} with $0$ being the time prior to the execution of the main loop. First, we make the following simple observation.

\begin{lemma}
\label{lemma:symmetry}
Consider \textsc{MinRanking} running on $H_{n,k}$. Then $\forall t \in \{0, 1, \ldots, n\}$ the subgraph remaining at time $t$ in \textsc{MinRanking} (disregarding isolated vertices) is isomorphic to $H_{n-t, k-k'}$ for some $k'\le k$. In particular, $\forall t \in \{0,1,\ldots, n\}$ every unmatched vertex in $U$ has the same current degree at time $t$.
\end{lemma}
\begin{proof}
This follows by a simple induction on $t$. Isolated vertices do not affect the performance of \textsc{MinRanking}, so we freely disregard them whenever they occur. Case $t=0$ is trivial. By the induction hypothesis, at time $t$ we have a graph isomorphic to $H_{n-t,k-k'}$. In the inductive step, there are possibly two cases at time $t+1$. In the first case, a vertex in $U$ is matched with a vertex in $V_2$. After deleting both of these vertices, the remaining subgraph is isomorphic to $H_{n-t-1,k-k'} = H_{n-(t+1),k-k'}$. In the second case, a single vertex $u$ in $U$ is matched with a vertex in $V_1$ (if there is one remaining, i.e., if $k-k'\ge 1$). After deleting both of these vertices, the unique vertex in $V_2$ that used to be the neighbor of $u$ becomes isolated. Disregarding this isolated vertex, the remaining subgraph is isomorphic to $H_{n-t-1,k-k'-1}=H_{n-(t+1),k-(k'+1)}$.
\end{proof}
Therefore, by symmetry, we may assume from now on that \textsc{MinRanking} processes the vertices in $U$ in the order from $1_U$ to $n_U$. This assumption does not affect the expected number of $(U,V_2)$ edges in a matching output by \textsc{MinRanking}. 

\begin{definition}
Let \textsc{RHSGreedy} be the algorithm that receives vertices from $V_1 \cup V_2$ in a random order and matches the received vertex to the first available vertex in $U$, if there is one.
\end{definition}

It is easier to analyze the expected number of $(U,V_2)$ edges for \textsc{RHSGreedy} than for \textsc{MinRanking}. Fortunately, the two algorithms produce the same outputs.

\begin{lemma}
\label{lemma:equivalence}
For a given fixed permutation $\pi$ of vertices in $V_1 \cup V_2$ the matching constructed by \textsc{MinRanking} is exactly the same as the matching constructed by \textsc{RHSGreedy} on $H_{n,k}$. Consequently, the expected number of $(U,V_2)$ edges in the output is the same for the two algorithms.
\end{lemma}
\begin{proof}
Let $v_1, \ldots, v_{n+k}$ be the vertices in $V_1 \cup V_2$ listed in the order given by $\pi$. By Lemma~\ref{lemma:symmetry}, we may assume that the order in which \textsc{MinRanking} considers vertices in $U$ is given by $1_U, \ldots, n_U$. Note that \textsc{RHSGreedy} has $n+k$ iterations, since in iteration $i$ it considers $v_i$, while \textsc{MinRanking} has $n$ iterations, since in iteration $i$ it considers $i_U$. Thus, the two algorithms run in different time domains. However, if we ignore iterations in which \textsc{RHSGreedy} considers isolated vertices, it is easy to see that both algorithms have $n$ iterations. We prove that the matching constructed by \textsc{MinRanking} is the same as the matching constructed by \textsc{RHSGreedy} by induction on the number of  $U$ vertices matched  by \textsc{MinRanking}.

Let $v_j$ denote the vertex matched by \textsc{MinRanking} with $1_U$. The vertex $1_U$ can be matched either with a vertex in $V_1$ or with $1_{V_2}$. We consider these two cases separately. 

Suppose that $v_j \in V_1$. By the definition of \textsc{MinRanking}, there exists $k >j$ such that $1_{V_2} = v_k$, and for all $i < j$ we have $v_i \in V_2$. Therefore, \textsc{RHSGreedy} on this input matches each $v_i \in V_2$ with its unique neighbor in $U$ for $i < j$. Then $v_j$ will be matched with $1_U$, as it will be the first available neighbor of $v_j$ in $U$. Note that removing the two matched vertices together with $1_{V_2}$ (which becomes isolated) does not affect the matchings constructed by the two algorithms on the resulting subgraphs. Thus, we can eliminate these vertices and proceed inductively.

Now, suppose that $v_j \in V_2$. By the definition of \textsc{MinRanking}, for all $i < j$ we have $v_i \in V_2$. Arguing as above, \textsc{RHSGreedy} on this input matches each $v_i \in V_2$ with its unique neighbor in $U$ for $i \le j$. In particular, $v_j$ is matched with $1_U$. Eliminating the two matched vertices does not affect the rest of the matching constructed by the two algorithms. Thus, we can proceed inductively.

We showed that in both cases both algorithms include the same edge $(1_U, v_j)$ in their outputs, and we can proceed inductively after removing $1_U, v_j$ and possibly $1_{V_2}$. 
\end{proof}

\begin{lemma}
\label{lemma:markov}
Consider \textsc{RHSGreedy} running on $H_{n,n}$. Let $X_n(t)$ denote the number of vertices in $U$ unmatched at iteration $t$ (ignoring trivial iterations with isolated vertices). Let $Y_n(t)$ denote the number of vertices in $V_2$ that have been matched up to and including iteration $t$ (ignoring trivial iterations with isolated vertices). Then $(X_n(t),Y_n(t))$ is a Markov chain with the state space $\{0,1,\ldots, n\}^2$ and the following transition probabilities:
\[ P((X_n(t+1),Y_n(t+1))=(x-1,y) \mid (X_n(t),Y_n(t))=(x,y)) = \frac{x+y}{2x+y},\text{ and}\]
\[ P((X_n(t+1),Y_n(t+1))=(x-1,y+1) \mid (X_n(t),Y_n(t))=(x,y)) = \frac{x}{2x+y},\]
where $ t \in \{0, 1, \ldots, n-1\}$. The initial state is $(X_n(0),Y_n(0)) = (n,0)$.
\end{lemma}
\begin{proof}
Let $X_n(t) = x$ and $Y_n(t) = y$. Then the number of vertices from $V_2$ that remain as possible candidates for a match with a vertex in $U$ is $x$. The number of vertices that have been matched in $V_1$ is $n-x-y$. Thus, the number of vertices in $V_1$ that remain as possible candidates for a match with a vertex in $U$ is $n-(n-x-y) = x+y$. In the next nontrivial iteration, the number of remaining $U$ vertices goes down by one, hence $X_n(t+1) = x-1$. The number of matched vertices in $V_2$ can either remain the same or increase by one. \textsc{RHSGreedy} matches a vertex from $V_2$ with a vertex in $U$ if and only if a vertex from $V_2$ appears next in the permutation of the remaining vertices from $V_1 \cup V_2$, which happens with probability $x/(2x+y)$. Therefore, we have
\[ P((X_n(t+1),Y_n(t+1))=(x-1,y+1) \mid (X_n(t),Y_n(t))=(x,y)) = \frac{x}{2x+y}.\]
The probability of the other case, i.e., $Y_n(t+1) = y+1)$, is simply $1 - x/(2x+y) = (x+y)/(2x+y)$.
\end{proof}

By Lemma~\ref{lemma:equivalence}, the expected number of $(U,V_2)$ edges in a matching constructed by \textsc{MinRanking} is equal to $\mathbb{E}(Y_n(n))$. The asymptotic behavior of $\mathbb{E}(Y_n(n))$ is captured by a solution to a related ordinary differential equation -- this is an application of the differential equation method (see \cite{Kurtz70}, \cite{Wormald1995}, \cite{Wormald1999}). We get the following result.

\begin{lemma}
\label{lemma:helper_for_main_theorem}
Let the setting be as in Lemma~\ref{lemma:markov}. For every $\epsilon > 0$ we have
\[ (1-\epsilon)\frac{n}{e} + o(n) \le \mathbb{E}(Y_n(n)) \le \frac{n}{e} + o(n).\]
We remark that $\mathbb{E}(Y_n(n))$ is also equal to the expected number of $(U,V_2)$ edges in a matching constructed by \textsc{MinRanking}.
\end{lemma}
\begin{proof}
Define $\Delta Y_n(t)  = Y_n(t) - Y_n(t-1)$ and $\Delta X_n(t) = X_n(t)-X_n(t-1)$. Then by Lemma~\ref{lemma:markov} we have
\[ \mathbb{E}(\Delta Y_n(t)) = \frac{X_n(t)}{2 X_n(t) + Y_n(t)}, \text{ and}\]
\[ \mathbb{E}(\Delta X_n(t)) = -1.\]
Thus, we have
\[ \frac{\mathbb{E}(\Delta Y_n(t))}{ \mathbb{E}(\Delta X_n(t))} = \frac{-X_n(t)}{2 X_n(t) + Y_n(t)}.\]
Applying the differential equation method (see \cite{Kurtz70}, \cite{Wormald1995}, \cite{Wormald1999}), the asymptotics of $Y_n(n)$ are captured by a solution to the following ordinary differential equation ($y$ is regarded as a differentiable function of $x$):
\[ \frac{d y}{d x} = \frac{-x}{2x+y}\]
with the initial condition $y(n) = 0$. Let $\tilde{y}(x)$ be defined implicitly as a solution to the following equation:
\[ \ln(1+\tilde{y}(x)/x) - \frac{1}{1+\tilde{y}(x)/x} = -1 + \ln(n)-\ln(x).\]
Then it is straightforward to verify that $\tilde{y}(x)$ is a solution to the above ODE. Note that the function $\tilde{y}(x)$ is undefined at $x=0$ (this essentially comes from the condition of Lemma~\ref{lemma:markov} that $t \le n-1$). However, we have $|\mathbb{E}(Y_n(n)) - \mathbb{E}(Y_n(n-1))| \le 1$, thus we can estimate the asymptotics of $\mathbb{E}(Y_n(n))$ by the asymptotics of $\mathbb{E}(Y_n(n-1))$, which amounts to finding the value of $\tilde{y}$ at $x = 1$:
\[  \ln(1+\tilde{y}(1)) - \frac{1}{1+\tilde{y}(1)} = -1 + \ln(n).\]
Define
\[ f(z) = \ln(1+z) - \frac{1}{1+z} +1 - \ln(n).\]
We have
\begin{align*}
f(n/e) &= \ln(1+n/e) - \ln(n/e) - \frac{1}{1+n/e}\\
&= \ln(1+e/n) - \frac{1}{1+n/e}\\
&= \left(\frac{e}{n}-\frac{e^2}{2n^2}+o(n^{-3})\right) - \left( \frac{e}{n}-\frac{e^2}{n^2}+o(n^{-3})\right)\\
&=\frac{e^2}{2n^2}+o(n^{-3})\\
&> 0\\
\end{align*}
where the third equality follows from the Taylor series for $\ln(1+x)$ and $1/(1+1/x)$ around 0, and the last inequality holds for large enough $n$.

Similarly, we have
\begin{align*}
f((1-\epsilon)n/e) &= \ln(1+(1-\epsilon)n/e) - \ln(n/e) - \frac{1}{1+(1-\epsilon)n/e}\\
&= \ln(1-\epsilon+e/n) - \frac{1}{1+(1-\epsilon)n/e}\\
&\le \ln(1-\epsilon/2) - \frac{1}{1+(1-\epsilon)n/e}\\
&\le -\epsilon/2 - o(n^{-1})\\
&<0\\
\end{align*}
where the first inequality holds for large enough $n$ since $\ln$ is an increasing function, the second inequality follows from $e^x \ge 1+x$, and the last inequality holds for large enough $n$.

Since $f(n/e) > 0$ and $f((1-\epsilon)n/e) < 0$, by the intermediate value theorem $f(z)$ has a root in the interval $((1-\epsilon)n/e+o(n), n/e+o(n))$. This is precisely the statement of this lemma.
\end{proof}

Now we are in a position to prove the main theorem of this section.

\begin{proof}[Proof of Theorem~\ref{theorem:main_theorem}] Fix $\epsilon > 0$. Consider \textsc{MinRanking} running on graph $G_b$. We introduce the following variables to denote the number of edges of certain types  in a matching constructed by \textsc{MinRanking}:
\begin{itemize}
\item $X_i$  denotes the number of $(S_{2,L}^{(i)},S_{1,R})$ edges,
\item $Y_i$ denotes the number of $(S_{2,L}^{(i)},S_{2,R})$ edges,
\item $Z$ denotes the number of $(S_{1,L},S_{2,R})$ edges.
\end{itemize}

\textsc{MinRanking} first matches nodes in $S_{2,L}$ since these nodes have the minimum degree of $b+1$. Once \textsc{MinRanking} matches a node in some block $S_{2,L}^{(i)}$, it will continue matching the nodes in the same block until that block is exhausted. The algorithm then moves on to the next block. This continues for at least $b-1$ blocks. In the unlikely scenario that all $b^2-b$ nodes of $S_{1,R}$ have been matched in this process, nodes from $S_{3,L}$ would have degree $b+1$. Therefore, after considering at least $b-1$ blocks from $S_{2,L}$ the algorithm might start matching nodes from other parts of $L$, but not sooner. The subgraph induced by $S_{2,L}^{(i)} \cup S_{2,R}^{(i)} \cup S_{1,R}^{(i)}$ is $H_{b,b}$, where $S_{1,R}^{(i)}$ consists of the nodes connected in parallel to $S_{2,L}^{(i)}$. Therefore, by Lemma~\ref{lemma:helper_for_main_theorem} we have
\begin{equation}\label{equation:expxi}
b^2/e+o(b^2) \ge \mathbb{E}\left[\sum_{i=1}^b X_i \right] \ge (1 - \epsilon) b^2/e + o(b^2)
\end{equation}
The algorithm will match all the vertices from the first $b-1$ blocks of $S_{2,L}$. Therefore, we have
\begin{equation}\label{equation:xisumyi}
b^2 \ge \sum_{i=1}^b X_i + Y_i \ge (b-1)b.
\end{equation}
Note that $X_i$ also counts the number of nodes from $S_{2,R}^{(i)}$ available to be matched with $S_{1,L}$. Therefore, disregarding nodes from $S_{3,L}$ the number of nodes matched between $S_{1,L}$ and $S_{2,R}$ is at least $\left(\sum_{i=1}^b X_i \right) - 2b$, i.e., we have
\begin{equation}\label{equation:zlb}
\left(\sum_{i=1}^b X_i \right) \ge Z \ge \left(\sum_{i=1}^b X_i \right) - 2 b.
\end{equation}
Putting all this together, the expected size of the matching output by \textsc{MinRanking} is
\begin{align*}
\mathbb{E} \left[ Z + \sum_{i=1}^b X_i+Y_i \right] &\ge \mathbb{E}[Z] + (b-1)b\\ 
&\ge \mathbb{E} \left[\sum_{i=1}^b X_i\right] -2b +(b-1)b \\
&\ge (1-\epsilon)b^2/e+b^2+o(b^2) \\
&= b^2(1+(1-\epsilon)/e) + o(b^2),
\end{align*}
where the first inequality is by~\eqref{equation:xisumyi}, the second inequality is by~\eqref{equation:zlb}, and the last inequality is by~\eqref{equation:expxi}. Similarly, we have
\[ \mathbb{E} \left[ Z + \sum_{i=1}^b X_i+Y_i \right] \le b^2(1+1/e)+o(b^2).\]
Since the graph has a perfect matching of size $2b^2+2b$, the asymptotic approximation ratio of \textsc{MinRanking} on $G_b$ is at least $1/2 + 1/(2e)$ and at most $1/2+(1-\epsilon)/(2e)$. Since $\epsilon > 0$ is arbitrary, the theorem follows.
\end{proof}

The expectation $\mathbb{E}_{\sigma, \pi} [\textsc{Ranking}(\sigma,\pi)]$  in the analysis of \textsc{MinRanking} is with respect to the joint 
distribution for the \textsc{Ranking} permutation $\pi$ of the offline nodes and the 
distribution on $\sigma$ induced by the tie breaking rule for the online nodes. 
The expectation can be written as $\mathbb{E}_{\sigma} [\mathbb{E}_{\pi} [\textsc{Ranking}(\sigma,\pi)| \sigma]]$, the expectation with respect to $\pi$ conditioned on the instantiations 
of $\sigma$. By averaging this implies that there is some fixed setting of $\sigma$ for which the inappropximation result holds. In this regard, our inapproximation for \textsc{MinRanking} is the analogue   
of Besser and Poloczek's priority analysis of the \textsc{MinGreedy} algorithm. 
Both results show that what seems to be a ``well-motivated'' deterministic 
(or random) ordering of the 
online vertices performs worse than a naive uniform random 
ordering of the
online vertices.

\section{Conclusion and Open Problems}
\label{sec:conclusion}
We have considered a number of ``online-based'' algorithms for the maximum bipartite matching problem. We believe that the algorithms considered in this paper all pass the intuitive ``you know it when you see it'' standard for conceptually simple algorithms. In particular, these algorithms take linear time in the number of edges  and are very easy to implement. Even given the restricted nature of these algorithms, it is a challenge to understand their performance. 

Our results for the MBM, in conjunction with the results in Poloczek \etal.~\cite{PoloczekSWZ17} for MaxSat show both the promise and limitations of conceptually simple algorithms. Many open problems are suggested by this work. Clearly, any problem studied in the competitive online literature can be considered within the expanded framework of conceptually simple algorithms that in some way can be considered as expansions of the online model. In particular, for what problems is there a 
general method for de-randomizing online algorithms? Is there a precise algorithmic model that lends itself to 
analysis and captures multi-pass algorithms? 
And in addition to worst case and stochastic analysis, how would any of the conceptually simple MBM algorithms perform ``in practice''.  

There are a number of more specific immediate questions that deserve consideration. We collect a few of them here.

\begin{openproblem}
Can we establish a MBM inapproximation result for {\it any} category algorithm (i.e. a 
two pass deterministic algorithm that 
online assigns a category to each input item which is then used as the advice in the second pass)? 
In particular, is the D{\"u}rr \etal. Category-Advice algorithm the best such algorithm using say 
$O(1)$ or  
$O(\log n)$  bits per input?

\end{openproblem}

\begin{openproblem}
Does there exist a determnistic mutli-pass linear time algorithm that 
matches or improves upon the $1-1/e$ approximation of the \textsc{Ranking} algorithm?

\end{openproblem}

\begin{openproblem}
What is the exact approximation ratio achieved by \textsc{MinRanking}? What is the best achievable ratio by any randomized  priority algorithm?
Note that Pena and Borodin \cite{PenaB16} show that 
{\it every} deterministic priority algorithm for MBM cannot asymptotically 
exceed 
the $1/2$ approximation that holds for any greedy algorithm.   
\end{openproblem}

\begin{openproblem}
Find a conceptually simple tie-breaking rule for the greedy algorithm that beats the $1-1/e$ approximation ratio in the known IID model.
\end{openproblem}

\begin{openproblem} 
Is there a randomized priority algorithm for MBM that achieves an approximation ratio better than that of \textsc{Ranking} in the ROM model? 
\end{openproblem}

\begin{openproblem}
Is there a simple multi pass algorithm for the  vertex weighted matching problem where the offline vertices are weighted? Note that Aggarwal \etal. \cite{AggarwalGKM11} show that a ``perturbed'' version of 
KVV achieves a $1-1/e$ approximation for this vertex weighted case.  
\end{openproblem}

\bibliography{bipartite}{}
\bibliographystyle{plain}

\end{document}